\documentclass[pdflatex]{sn-jnl}

\usepackage{enumitem}

\usepackage{graphicx}
\usepackage{subcaption}
\usepackage{xcolor}
\usepackage{hyperref}

\usepackage{lipsum}

\usepackage{tikz}
\usetikzlibrary{arrows,shapes}
\usetikzlibrary{calc}

\usepackage{amssymb}
\usepackage{stackrel}
\usepackage{mathtools}

\usepackage[dvips]{epsfig}
\usepackage[utf8]{inputenc}

\usepackage{multirow}%
\usepackage{amsmath,amssymb,amsfonts}%
\usepackage{amsthm}%
\usepackage{mathrsfs}%
\usepackage[title]{appendix}%
\usepackage{xcolor}%
\usepackage{textcomp}%
\usepackage{manyfoot}%
\usepackage{booktabs}%
\usepackage{listings}%

\usepackage{bbm}
\usepackage{mathtools}
\usepackage{bbold}

\theoremstyle{thmstyleone}%
\newtheorem{theorem}{Theorem}
\newtheorem{proposition}[theorem]{Proposition}%
\newtheorem{corollary}[theorem]{Corollary}%
\newtheorem{lemma}[theorem]{Lemma}%

\theoremstyle{thmstyletwo}%

\theoremstyle{thmstylethree}%
\newtheorem{definition}[theorem]{Definition}%

\theoremstyle{definition}
\newtheorem{remark}[theorem]{Remark}%

\begin{document}
\title[Morse Sequences and Reference Maps]{Morse Sequences: A simple approach to discrete Morse theory}

\author{\fnm{Gilles} \sur{Bertrand}}\email{gilles.bertrand@esiee.fr}

\affil{\orgdiv{LIGM}, \orgname{Univ Gustave Eiffel}, \orgaddress{\city{Marne-la-Vallée}, \postcode{F-77454}, \country{France}}}


\newcommand{\bbDelta}{%
  \Delta\mkern-12mu\Delta%
}

\newcommand{\bbNabla}{%
  \nabla\mkern-12mu\nabla%
}

\newcommand{\ELIMINE}[1]{}

\newcommand{\bbbu}{\; \ddot{\cup} \;}
\newcommand{\axr}[1]{\ddot{\textsc{#1}}  \normalsize}

\newcommand{\axcup}{\textsc{C\tiny{UP}} }
\newcommand{\axcap}{\textsc{C\tiny{AP}} }
\newcommand{\axunion}{\textsc{U\tiny{NION}} }
\newcommand{\axinter}{\textsc{I\tiny{NTER}} }

\newcommand{\bb}[1]{\mathbb{#1}}
\newcommand{\ca}[1]{\mathcal{#1}}
\newcommand{\ax}[1]{\textsc{#1} \normalsize}

\newcommand{\axb}[2]{\ddot{\textsc{#1}} \textsc{\tiny{#2}}  \normalsize}
\newcommand{\bbb}[1]{\ddot{\mathbb{#1}}}
\newcommand{\cab}[1]{\ddot{\mathcal{#1}}}

\newcommand{\rel}[1]{\scriptstyle{\mathbf{#1}}}

\newcommand{\rela}[1]{\textsc{\scriptsize{\bf{{#1}}}} \normalsize}

\newcommand{\de}[2]{#1[#2]}
\newcommand{\di}[2]{#1\langle #2 \rangle}

\newcommand{\la}{\langle}
\newcommand{\ra}{\rangle}

\newcommand{\cell}{\mathbb{C}}
\newcommand{\cellp}{\mathbb{C}^\times}
\newcommand{\simp}{\mathbb{S}}
\newcommand{\comp}{\mathbb{H}}
\newcommand{\simpp}{\mathbb{S}^\times}
\newcommand{\den}{\mathbb{D}\mathrm{en}}
\newcommand{\ram}{\mathbb{R}\mathrm{am}}
\newcommand{\tree}{\mathbb{T}\mathrm{ree}}
\newcommand{\graph}{\mathbb{G}\mathrm{raph}}
\newcommand{\vertex}{\mathbb{V}\mathrm{ert}}
\newcommand{\edge}{\mathbb{E}\mathrm{dge}}
\newcommand{\equ}{\mathbb{E}\mathrm{qu}}
\newcommand{\esub}{\mathbb{E}\mathrm{sub}}

\newcommand{\topp}{\langle \mathrm{K} \rangle}
\newcommand{\topq}{\langle \mathrm{Q} \rangle}
\newcommand{\topxp}{\langle \mathbb{X}, \mathrm{P} \rangle}
\newcommand{\topxq}{\langle \mathbb{X},\mathrm{Q} \rangle}

\newcommand{\vt}{\mathcal{K}}
\newcommand{\vtp}{\mathcal{T'}}
\newcommand{\vtpp}{\mathcal{T''}}
\newcommand{\vq}{\mathcal{Q}}
\newcommand{\vqp}{\mathcal{Q'}}
\newcommand{\vqpp}{\mathcal{Q''}}
\newcommand{\vk}{\mathcal{K}}
\newcommand{\vkp}{\mathcal{K'}}
\newcommand{\vkpp}{\mathcal{K''}}

\newcommand{\C}{\ensuremath{\searrow^{\!\!\!\!\!C}}}
\newcommand{\Detach}{\ensuremath{\;\oslash\;}}

\newcommand{\sig}{\sigma}
\newcommand{\del}{W}
\newcommand{\ms}{\overrightarrow{W}}
\newcommand{\msi}{\overrightarrow{W_{i}}}
\newcommand{\msim}{\overrightarrow{W_{i-1}}}
\newcommand{\mss}{\widehat{W}}
\newcommand{\msc}{\ddot{W}}

\newcommand{\pl}{\partial}
\newcommand{\plc}{\widehat{\partial}}
\newcommand{\dl}{\delta}
\newcommand{\dlc}{\widehat{\delta}}

\newcommand{\torusms}{\overrightarrow{T}}
\newcommand{\torusmso}{\overleftarrow{T}}

\newcommand{\corf}{\curlyvee}
\newcommand{\rf}{\curlywedge}

\newcommand{\ticorf}{\widetilde{\curlyvee}}
\newcommand{\tirf}{\widetilde{\curlywedge}}

\newcommand{\ddcorf}{\ddot{\nabla}}
\newcommand{\ddrf}{\ddot{\triangle}}

\newcommand{\docorf}{\dot{\nabla}}
\newcommand{\dorf}{\dot{\triangle}}



\newcommand{\ov}{\overline}
\newcommand{\un}{\underline}
\newcommand{\ha}{\widehat}
\newcommand{\ti}{\widetilde}
\newcommand{\dd}{\ddot}

\newcommand{\ka}{\kappa}
\newcommand{\ova}{\overrightarrow{w}}

\abstract{
In this paper, we develop the notion of a Morse sequence, offering an alternative approach to discrete Morse theory that is both simple and effective. A Morse sequence on a finite simplicial complex consists solely of two elementary operations: expansions (the inverse of collapses) and fillings (the inverse of perforations). Alternatively, a Morse sequence can be constructed using only collapses and perforations, providing a dual perspective. Such sequences serve as another representation of the gradient vector field of an arbitrary discrete Morse function.

To each Morse sequence, we associate a reference map and an extension map. The reference map assigns a set of critical simplices to each simplex in the complex, while the extension map assigns a set of simplices to each critical simplex. By considering the boundary of each critical simplex, these maps yield a chain complex that corresponds exactly to the Morse complex.

We demonstrate that, when restricted to homology, the extension map is the inverse of the reference map. Furthermore, these maps enable a direct derivation of the isomorphism theorem, which establishes the equivalence between the homology of a given object and that of its Morse complex.

Finally, we introduce the notion of an extension complex, defined exclusively in terms of extension maps. We prove that this concept is equivalent to the classical flow complex.

}


\maketitle

\section{Introduction} \label{sec1}

Discrete Morse theory is a combinatorial adaptation of classical Morse theory, applicable to a wide range of topological problems. 
Initially developed by Robin Forman \cite{For98,For02}, this theory was first introduced through discrete Morse  functions, 
which identify specific cells known as critical cells. These critical cells encapsulate the essential topological characteristics of a given object. 
In more recent treatments \cite{kozlov2007,koz20},  discrete Morse theory is presented through the notion of a combinatorial matching,
which is equivalent to the gradient vector field of a discrete Morse function.

In this paper \footnote{A preliminary version of some parts of this paper appeared in  \cite{Bertrand2023MorseSequences} and \cite{Bertrand2023MorseFrames}}, 
we propose an approach where, instead of a Morse function or a gradient vector field, a sequence of elementary operators 
is used for a simple representation of an object. This sequence, which we called \emph{a Morse sequence},
consists solely of two elementary operations:
expansions (the inverse of a collapse), and fillings (the inverse of a perforation).
These operations correspond precisely to those introduced by Henry Whitehead \cite{Whi39}.

After some basic definitions (Section \ref{sec:basic}), we introduce Morse sequences and give two meaningful examples (Sections \ref{sec:seq}).
We provide two computational schemes for extracting a Morse sequence from a given simplicial complex.
Also, we introduce a notion of equivalence between two Morse sequences, this notion is based on the \emph{gradient vector field} of a Morse sequence.

In section~\ref{sec:reference}, we present \emph{reference maps}, which are maps that associate a set of critical simplices to each simplex. These maps allow adding 
some crucial information to Morse sequences. In Section \ref{sec:grad}, we show that 
reference maps may be fully characterized by counting the number of gradient paths between each simplex and each critical face.

By simply considering the reference map of the boundary of each critical simplex, we obtain the \emph{critical complex} of a Morse sequence (Section \ref{sec:Mcomplex}),
which corresponds precisely
to the so-called  Morse complex.

In Section \ref{sec:reg}, we introduce \emph{arranged sequences} which are derived through a reordering of Morse sequences.
We also define the related notions of \emph{lower and upper skeletons} of a Morse sequence, which explicitly appear in arranged sequences.

Section \ref{sec:ext} introduces \emph{extension maps}, which associate a set of simplices to each critical simplex. 
We show that, when restricted to homology, an extension map is the inverse of a reference map. 
Furthermore, reference and extension maps allow us to directly recover an isomorphism between the homology of an object and the homology of its critical complex.

In Section \ref{sec:gflow}, we define the  \emph{extension complex}, which is based entirely on extension maps. 
We show that there is an isomorphism between the homology of an extension complex and the homology of a critical complex.
Moreover, we prove that this notion is equivalent to the classical flow complex.

Following the conclusion, we provide two appendices.
In the first one, we emphasize that 
a Morse sequence may
represent the gradient vector field of any arbitrary
discrete Morse function (Appendix \ref{app:vector}). In the second one, we make clear the relation between  Morse sequences and
different kinds of Morse
functions (Appendix~\ref{app:functions}).

Note that the paper is self-contained, the only external result which is used is Theorem \ref{th:gradi2} from  \cite{forman2002discretecohomology},
located at the very end of the paper. 

\section{Simplicial complexes, homology, and cohomology}
\label{sec:basic}
\subsection{Simplicial complexes}

Let $K$ be a finite collection of non-empty finite subsets of a set $S$. 
We say that $K$ is a
{\it(simplicial) complex}  if, for every set $\tau$  in $K$, and every non-empty subset $\sig \subseteq \tau$, the set $\sig$ also belongs to $K$.

An element of a simplicial complex $K$ is {\it a simplex of $K$}, or {\it a face of~$K$}.
A {\em facet of~$K$} is a simplex of $K$ that is maximal for inclusion.
The {\it dimension} of $\sigma \in K$, written $dim(\sigma)$,
is the number of its elements
minus one. If $dim(\sigma) =p$, we say that~$\sigma$ is a \emph{$p$-simplex}.
The {\it dimension of $K$}, written $dim(K)$,
is the largest dimension of its simplices,
the {\it dimension of $\emptyset$}, the void complex,  being defined to be $-1$. 

The {\em vertex set} of a simplicial complex $K$ is defined as the union of all faces of $K$. 
The elements of the vertex set are the {\em vertices of $K$}. For every vertex $v$ of $K$, 
the set~$\{ v \}$ is a face of $K$, and every face of $K$ is a subset of the vertex set.

We recall the definitions of simplicial collapses and simplicial expansions~\cite{Whi39}. \\
Let $K$ be a simplicial complex and let $\sig, \tau \in K$.
The couple $(\sig,\tau)$ is a {\em free pair for~$K$}, 
if $\tau$ is the only face of $K$ that contains $\sig$.
Thus $\tau$ is necessarily a facet of~$K$.
The \emph{dimension of a free pair $(\sig,\tau)$}
is defined by $dim(\sig,\tau) = dim(\tau)$.
We have  $dim(\sig,\tau) = dim(\sig) + 1$.
If $(\sig,\tau)$ is a free 
pair for $K$, then the simplicial complex $L = K \setminus \{ \sig,\tau \}$
is {\em an elementary 
collapse of $K$}, and $K$ is {\em an elementary 
expansion of $L$}. \\
We say that
$K$ {\em collapses onto $L$},
or that $L$ {\em expands onto $K$},
if there exists a  sequence 
$\langle K=K_0,\ldots,K_k=L \rangle$, such that
$K_i$ is an elementary collapse of $K_{i-1}$, $i \in [1,k]$; we say that such a sequence is a \emph{collapse sequence (from $K$ to $L$)}.
The complex
$K$ is {\em collapsible} if $K$ collapses onto a complex with a single vertex, that is, onto a complex of the form
$\{ \{ a \} \}$.

\subsection{Chain and cochain complexes}
\label{subsec:chain}

A \emph{chain complex}  $( C_p, d_p ) $ is a sequence of vector spaces $C_p$, with $p \in \mathbf{Z}$, 
connected by linear transformations $d_p : C_p \rightarrow C_{p-1}$, called \emph{boundary operators}
with the property that $d_{p-1} \circ d_{p} = 0$. Each element of $C_p$ is a \emph{$p$-chain}. \\
In a dual way, 
a \emph{cochain complex} $( C^p, d^p ) $ is a sequence of vector spaces $C^p$,  with $p \in \mathbf{Z}$, 
connected by linear transformations $d^p : C^p \rightarrow C^{p+1}$ called \emph{coboundary operators}
with the property that $d^{p+1} \circ d^p = 0$. Each element of $C^p$ is a \emph{$p$-cochain}. 

Let $( C_p, d_p )$ and $( C^p, d^p )$ be a chain and a cochain complex.
We define the four vector spaces:
\begin{itemize}
\item the set $Z_p(C)$ of \emph{$p$-cycles of $(C_p, d_p )$}, $Z_p(C)$ is the kernel of $d_p$;
\item the set $B_p(C)$ of \emph{$p$-boundaries of $(C_p, d_p )$}, $B_p(C)$ is the image of $d_{p+1}$;
\item the set $Z^p(C)$ of \emph{$p$-cocycles of $( C^p, d^p )$}, $Z^p(C)$ is the kernel of $d^p$;
\item the set $B^p(C)$ of \emph{$p$-coboundaries of $( C^p, d^p )$}, $B^p(C)$ is the image of $d^{p-1}$.
\end{itemize}

\noindent
By the definition of  (co)boundary operators, all (co)boundaries are (co)cycles.
We obtain homology and cohomology with the following quotient vector spaces:
\begin{itemize}
\item $H_p(C) = Z_p(C) / B_p(C)$, which is the \emph{$p^\textnormal{th}$ homology vector space of $(C_p, d_p )$};
\item $H^p(C) = Z^p(C) / B^p(C)$, which is the \emph{$p^\textnormal{th}$ cohomology vector space of $(C^p, d^p )$}.
\end{itemize}

\noindent
An element $h$ in $H_p(C)$  is such that $h = z + B_p(C)$ for some $z \in Z_p(C)$. We
write $h = [ z ]$, which  is the \emph{homology class of the cycle z}.

\noindent
Similarly, an element $h$ in $H^p(C)$  is such that $h = z + B^p(C)$ for some $z \in Z^p(C)$. We
also write $h = [ z ]$, which is the \emph{cohomology class of the cocycle z}.

\subsection{Homology and cohomology modulo 2}
\label{sec:hom}

Let $K$ be a simplicial complex. We write $K^{(p)}$ for the set 
of all $p$-simplices of~$K$, and 
$K[p]$ for the set composed of all subsets of $K^{(p)}$.
Thus $K^{(p)} = \emptyset$ and $K[p] = \{\emptyset\}$ whenever $p < 0$ or $p > dim(K)$.
Each element of $K[p]$ is a \emph{$p$-chain of $K$}.
The symmetric difference of two elements of $K[p]$ endows $K[p]$
with the structure of a vector space over the field $\mathbb{Z}_2 = \{ 0, 1 \}$. The set $K^{(p)}$
is a basis for this vector space.
Within this structure, a chain $c \in K[p]$ may be written as a sum
$\sum_{\sigma \in c} \sigma$, the chain $c = \emptyset$
being written $0$. The sum of two chains is obtained using the modulo 2 arithmetic.
See Chapter 8 of \cite{Gib10} and see \cite{Haus14} for a 
general 
presentation of modulo 2 homology. \\
If $S \subseteq K$, we write $S^{(p)} = \{\nu \in S \; | \; \nu \in K^{(p)} \}$ and
 $S[p] = \{c \subseteq S^{(p)} \}$. 

Let $\sigma \in K^{(p)}$, we set:

\noindent
\hspace*{\fill}
$\partial(\sigma) = \{\tau \in K^{(p-1)} \; | \; \tau \subset \sigma \}$ and
$\delta(\sigma) = \{\tau \in  K^{(p+1)}\; | \; \sigma \subset \tau \}$.
\hspace*{\fill}

\noindent
The \emph{boundary operator} $\partial_p : K[p] \rightarrow K[p-1]$ is such that, for
each $c \in K[p]$, 

\noindent
\hspace*{\fill}
$\partial_p(c) = \sum_{\sigma \in c} \partial(\sigma)$, with $\partial_p(\emptyset) = 0$.
\hspace*{\fill}

\noindent
The \emph{coboundary operator} $\delta^p : K[p] \rightarrow K[p+1]$ is such that, for
each $c \in K[p]$,

\noindent
\hspace*{\fill}
$\delta^p(c) = \sum_{\sigma \in c} \delta(\sigma)$, with $\delta^p(\emptyset) = 0$.
\hspace*{\fill}

For each $p \in \mathbf{Z}$, we have
$ \partial_{p} \circ \partial_{p+1} = 0$ and $\delta^{p+1} \circ \delta^{p} = 0$. 
Thus, $(C_p, d_p ) = ( K[p], \pl_p )$ is a chain complex and $( C^p, d^p ) = ( K[p], \dl^p )$ is a cochain complex.
We write  $Z_p(K)$, $B_p(K)$, $H_p(K)$ for $Z_p(C)$, $B_p(C)$, $H_p(C)$, and $Z^p(K)$, $B^p(K)$, $H^p(K)$ for $Z^p(C)$, $B^p(C)$, $H^p(C)$.

Let $\beta_p(K) = dim(H_p(K))$ and $\beta^p(K) = dim(H^p(K))$.
We have $\beta_p(K) = \beta^p(K)$
(See \cite[Sec. V.1]{Edel01}).
The number $\beta_p(K) = \beta^p(K)$ is the \emph{$p^\textnormal{th}$ Betti number (mod 2) of $K$}.
Informally, the $p^\textnormal{th}$ Betti number of $K$ corresponds to the number of ``$p$-dimensional holes'' of the complex $K$.

\def\figsize1{0.22}
\begin{figure*}[tb]
    \centering
        \begin{subfigure}[t]{\figsize1\textwidth}
        \centering
        \includegraphics[width=.99\textwidth]{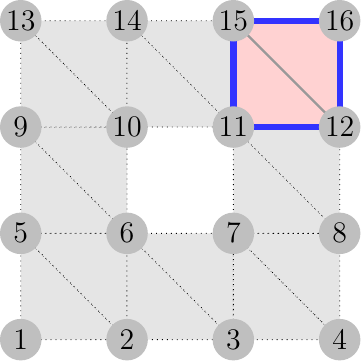}
        \caption{%
        }
    \end{subfigure}%
    ~
    \begin{subfigure}[t]{\figsize1\textwidth}
        \centering
        \includegraphics[width=.99\textwidth]{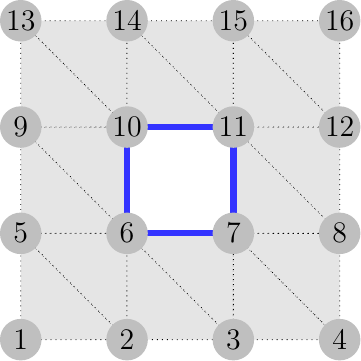}
        \caption{%
        }
    \end{subfigure}%
    ~
    \begin{subfigure}[t]{\figsize1\textwidth}
        \centering
        \includegraphics[width=.99\textwidth]{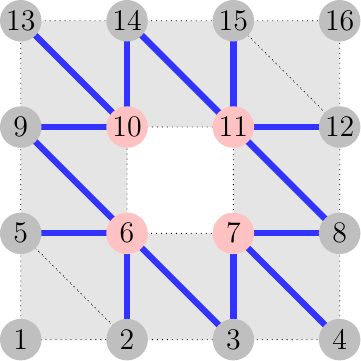}
        \caption{%
        }
    \end{subfigure}%
    ~
    \begin{subfigure}[t]{\figsize1\textwidth}
        \centering
        \includegraphics[width=.99\textwidth]{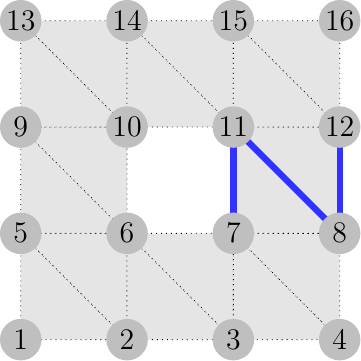}
        \caption{%
        }
    \end{subfigure}
    \caption{\label{fig:annulus} An annulus, with various cycles and cocyles. See text for details. %
}
\end{figure*}

Fig.~\ref{fig:annulus} depicts an annulus, with various cycles and cocyles, coloured in blue. In Fig.~\ref{fig:annulus}.a, we see a 1-cycle that is the 1-boundary of the two pink triangles.  In Fig.~\ref{fig:annulus}.b, we have a 1-cycle that is not a 1-boundary. Such a cycle detects a ``hole'' by ``contouring'' it.
In Fig.~\ref{fig:annulus}.c, we see
 a 1-cocycle which is the 1-coboundary of the four pink points. In Fig.~\ref{fig:annulus}.d, we have a 1-cocycle that is not a 1-coboundary. Such a cocycle detects a ``hole'' by ``cutting'' the annulus.


\section{Morse sequences}
\label{sec:seq}
Let us  start first with the definition of perforations and fillings.

Let $K,L$ be simplicial complexes.
If $\nu \in K$ is a facet of $K$ and if $L = K \setminus \{\nu \}$, we say that
$L$ is {\em an elementary perforation of $K$}, and that
$K$ is {\em an elementary filling of $L$}.

These transformations were introduced by Whitehead in a seminal paper \cite{Whi39}. Combined with collapses and expansions, it has been shown that we obtain four operators
that correspond to the homotopy equivalence between two simplicial complexes  (Th. 17 of \cite{Whi39}).
See also \cite{Ber21} which provides another kind of equivalence based on a variant of these operators.

We introduce the notion of a ``Morse sequence'' by simply considering expansions and fillings
of a simplicial complex. \\

\begin{definition} \label{def:seq1}
Let $K$ be a simplicial complex. A \emph{Morse sequence (on $K$)} is a sequence
$\ms = \langle \emptyset = K_0,\ldots,K_k =K \rangle$ of simplicial complexes such that,
for each $i \in [1,k]$, $K_i$ is either an elementary expansion or an elementary filling of $K_{i-1}$.
We also write $\ms(K)$ for a Morse sequence $\ms$ on $K$. \\
\end{definition}

\noindent
Let $\ms = \langle K_0,\ldots,K_k \rangle$ be a Morse sequence.
For each $i \in [1,k]$, let $\ka_i$ be such that:
\begin{itemize}
\item If $K_i$ is an elementary filling of $K_{i-1}$ and $K_i = K_{i-1} \cup \{\nu \}$, then $\ka_i = \nu$. We say that
the face $\ka_i$ is \emph{critical for $\ms$}. 
\item If $K_i$ is an elementary expansion of $K_{i-1}$ and $K_i = K_{i-1} \cup \{\sigma,\tau \}$, with $\sig \subseteq \tau$, then
$\ka_i = (\sigma,\tau)$. The pair 
$\ka_i$ is \emph{regular for $\ms$}, the face $\sig$ is \emph{lower regular for $\ms$}, and the face $\tau$ is \emph{upper regular for $\ms$}.
\end{itemize}

\noindent
We write 
$\diamond \ms = \langle \ka_1, \ldots, \ka_k \rangle$, and we say that
$\diamond \ms$ is a \emph{simplex-wise (Morse) sequence}.

\noindent
We will use also the four following notations: 
\begin{itemize}
  \item $\widehat{W} = \{\nu \in  K \; | \; \nu$ is critical for $\ms \}$,
\item $\msc = \{(\sig,\tau) \in K \times K \; | \; (\sig,\tau)$ is a regular pair for $\ms \}$, 
\item $\overline{W} = \{\tau \in  K \; | \; \tau$ is upper regular for $\ms \}$, 
\item $\underline{W} = \{\sig \in  K \; | \; \sig$ is lower regular for $\ms \}$. 
\end{itemize}
Thus, $\diamond \ms$ is a sequence of faces of $\ha{W}$ and pairs of $\msc$. 
Clearly, $\ms$ and $\diamond \ms$  are two equivalent forms. 
That is, one form is determined by the other.
Also, we observe that $K = \ha{W} \cup \underline{W} \cup \overline{W}$. This partition of a complex 
corresponds to a classification  which is often associated to a matching in the Hasse diagram of a complex or a poset, see for example 
Section 10.2 of \cite{koz20}. 

\ELIMINE{
We write: \\
- $\eta(\sig) = \tau$ and $\eta(\tau) = \sig$ whenever $(\sig, \tau)$ is a regular pair for $\ms$. \\
- $\eta(\nu) = 0$ whenever $\nu$ is critical for $\ms$. \\
Thus, these expressions define a bijection $\eta : K \rightarrow K$ such that $\eta \circ \eta = Id_K$. \\
}

Observe that, if $\ms = \langle K_0,...,K_k \rangle$ is a Morse sequence, with $k \geq 1$, then
$K_1$ is necessarily a filling of $\emptyset$. 
Thus, $K_1$ is necessarily a vertex. That is, $K_1$ is made of a single $0$-simplex that is critical for $\ms$.

There are several ways to obtain a Morse sequence $\ms$ from a given complex~$K$.
The two following schemes are basic ones to achieve this goal, see \cite{Ben14,Mro09,Har14,fugacci2019computing,Dlotko11}
for similar schemes given in the context of discrete Morse theory:
\begin{enumerate}
\item \emph{The increasing scheme}. We build $\ms$ from the left to the right. Starting from~$\emptyset$, we obtain $K$ by iterative expansions and fillings.
We say that this scheme is \emph{maximal} if
we make a filling only if no expansion can be made.
\item \emph{The decreasing scheme}. We build $\ms$ from the right to the left. Starting from $K$, we obtain $\emptyset$ by iterative collapses and perforations.
We say that this scheme is \emph{maximal} if we make a perforation only if no collapse can be made.
\end{enumerate}
\begin{figure*}[tb]
    \centering
    \begin{subfigure}[t]{0.31\textwidth}
        \centering
        \includegraphics[height=.99\textwidth]{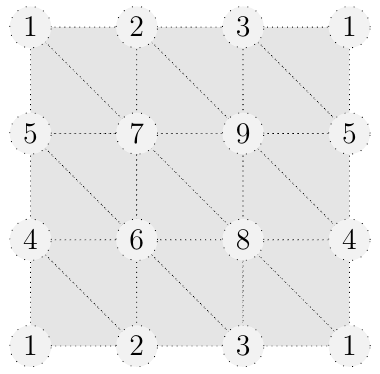}
        \caption{}
    \end{subfigure}%
    ~
    \begin{subfigure}[t]{0.31\textwidth}
        \centering
        \includegraphics[height=.99\textwidth]{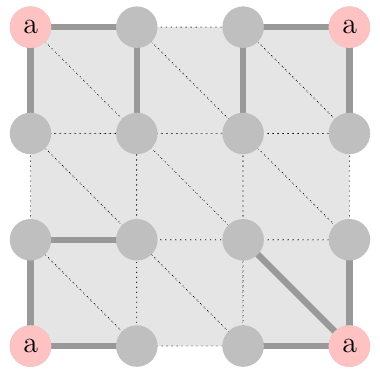}
        \caption{}
    \end{subfigure}%
    ~
    \begin{subfigure}[t]{0.31\textwidth}
        \centering
        \includegraphics[height=.99\textwidth]{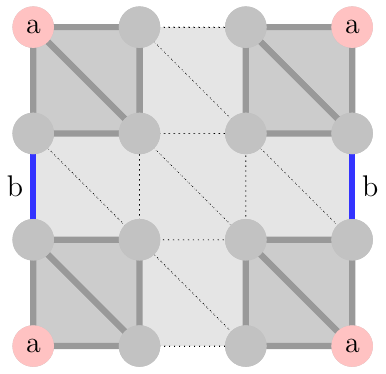}
        \caption{}
    \end{subfigure}%
    
    \begin{subfigure}[t]{0.31\textwidth}
        \centering
        \includegraphics[height=.99\textwidth]{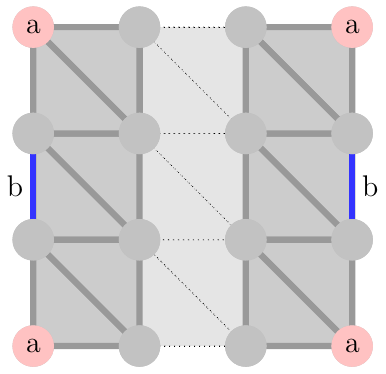}
        \caption{}
    \end{subfigure}%
    ~
    \begin{subfigure}[t]{0.31\textwidth}
        \centering
        \includegraphics[height=.99\textwidth]{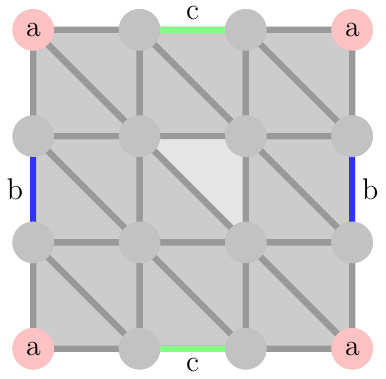}
        \caption{}
    \end{subfigure}
    ~
    \begin{subfigure}[t]{0.31\textwidth}
        \centering
        \includegraphics[height=.99\textwidth]{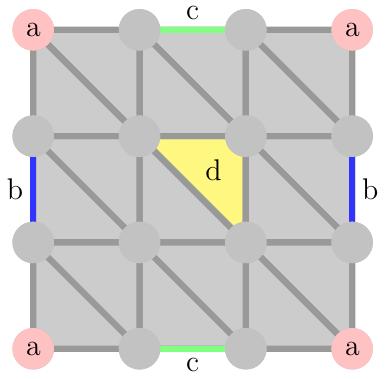}
        \caption{}
    \end{subfigure}

 \caption{Illustration of some steps of a Morse sequence on the torus. 
 (a) A triangulation, points with the same label are identified. 
 (b) The sequence begins with the critical 0-simplex~$a$,  then few elementary expansions are added to the sequence.
 (c) A maximal expansion from $a$ is done, then the critical 1-simplex $b$ is added to the sequence. 
 (d) A maximal expansion from $b$ is done. 
 (e) The second critical 1-simplex $c$ is added, and a maximal expansion from $c$ is done. (f) The critical 2-simplex~$d$ is added. }
 \label{fig:MorseSequenceTorus}
\end{figure*}
\begin{figure*}[tb]
    \centering
    \begin{subfigure}[t]{0.32\textwidth}
        \centering
        \includegraphics[width=.99\textwidth]{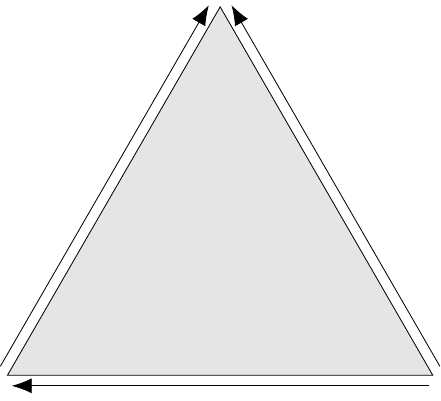}
        \caption{}
    \end{subfigure}%
    ~
    \begin{subfigure}[t]{0.32\textwidth}
        \centering
        \includegraphics[width=.99\textwidth]{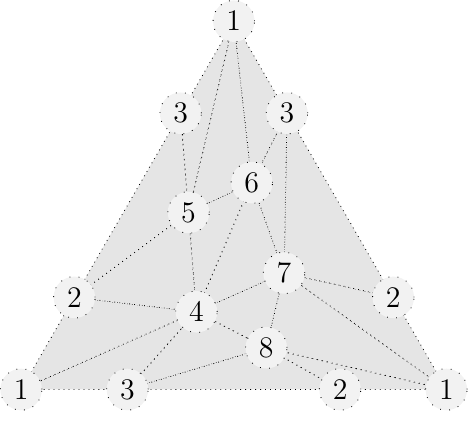}
        \caption{}
    \end{subfigure}%
    ~
    \begin{subfigure}[t]{0.32\textwidth}
        \centering
         \includegraphics[width=.99\textwidth]{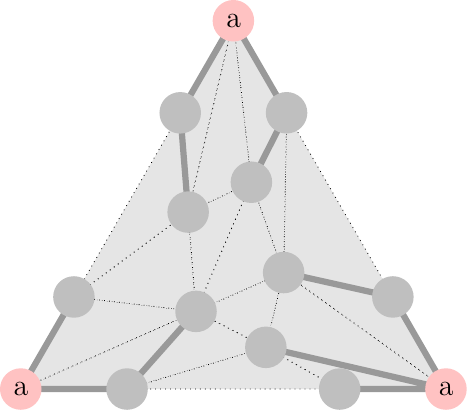}
        \caption{}
    \end{subfigure}%

    \begin{subfigure}[t]{0.32\textwidth}
        \centering
        \includegraphics[width=.99\textwidth]{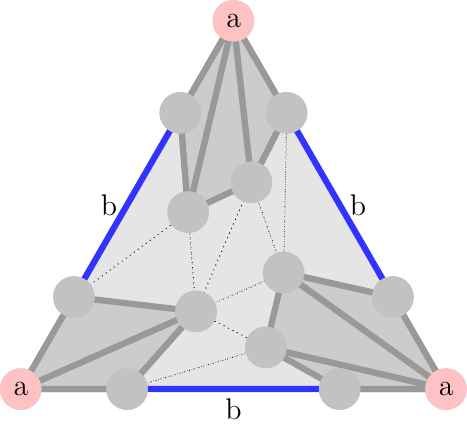}
        \caption{}
    \end{subfigure}%
    ~
    \begin{subfigure}[t]{0.32\textwidth}
        \centering
        \includegraphics[width=.99\textwidth]{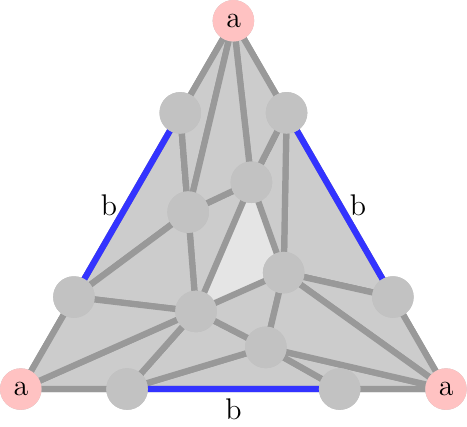}
        \caption{}
    \end{subfigure}
     ~
    \begin{subfigure}[t]{0.32\textwidth}
        \centering
        \includegraphics[width=.99\textwidth]{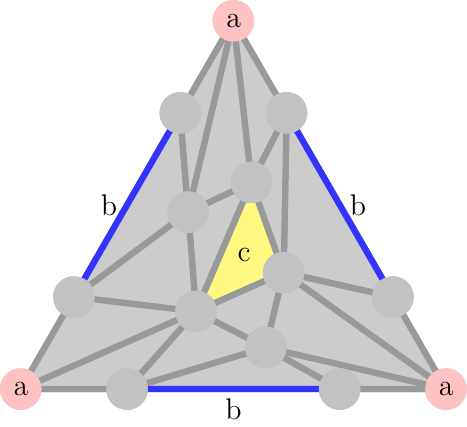}
        \caption{}
    \end{subfigure}

 \caption{Illustration of some steps of a Morse sequence on the dunce hat.  
 (a) the dunce hat, the three edges of the triangle have to be identified with the arrows. 
 (b) A triangulation of the dunce hat. 
 (c) The sequence begins with the critical 0-simplex $a$, then few elementary expansions are added to the sequence.
 (d) A maximal expansion from $a$ is done, then the 1-critical simplex $b$ is added. (e) A maximal expansion from~$b$. (f) The critical 2-simplex
 $c$ is added.}
  \label{fig:MorseSequenceDunceHat}
\end{figure*}
Fig.~\ref{fig:MorseSequenceTorus} presents an example of a Morse sequence  on a torus. The sequence has been obtained by following the maximal increasing scheme.
Using the same scheme, an example of a Morse sequence on the dunce hat  is given Fig.~\ref{fig:MorseSequenceDunceHat}.
The dunce hat is a well-known example of a contractible but not collapsible object, see \cite{Zee64}. 
 Contractibility can be seen by noting that the dunce hat embeds in a 3-ball which both collapses onto a vertex and onto the dunce hat. 
 We observe that: \\
 - The Morse sequence on the torus leads to two critical $1$-faces and one critical $2$-face. These numbers correspond exactly to the 
 numbers of 1D and 2D ``holes'' of the torus. \\
 - The dunce hat has no 1D or 2D ``hole'', nevertheless the Morse sequence leads to one critical $1$-face and one critical $2$-face. These numbers are the smallest we can obtain. \\
 In this paper, we use the torus and the dunce hat as our canonical examples, as they exemplify two contrasting cases which arise in the context of discrete Morse theory.\\
 

\begin{remark} \label{rem:NP}
The above maximal increasing and decreasing schemes are two methods 
which try to minimize the number of critical simplices of a Morse sequence. 
By the results given in \cite{Jos06,lewiner03}, this problem is NP-hard.
Therefore, these methods do not, in general, give optimal results. \\
\end{remark}

\begin{remark} \label{rem:col}
Let $\ms = \langle K_0,\ldots,K_k \rangle$ be a Morse sequence, let  $\diamond \ms = \langle \ka_1, \ldots, \ka_k \rangle$,
and let $\ka_i$, $\ka_j$, $j >i$, be two consecutive
critical faces of $\diamond \ms$, that is, $\ka_{i+1},...,\ka_{j-1}$ are regular pairs.
Then, as a direct consequence of the definition of a Morse sequence, the complex $K_{j-1}$ collapses onto $K_{i}$.
This property is the core of a fundamental theorem, called \emph{the collapse theorem}, which makes the link between the basic definitions of  discrete Morse theory
and discrete homotopy (See Theorem 3.3 of \cite{For98a} and Theorem 4.27 of \cite{Sco19}).
In a certain sense, we can say that Morse sequences provide an introduction to discrete Morse theory by starting from this property.\\
\end{remark}

\begin{remark} \label{rem:fil}
Any Morse sequence $\ms$ on $K$ is a \emph{filtration on $K$}, that is a sequence of nested complexes  $\langle \emptyset = K_0,...,K_k =K \rangle$
such that, for $i \in [0,k-1]$, we have $K_i \subseteq K_{i+1}$; see \cite{DW22}.
Also any \emph{simplex-wise filtration on $K$} is a special case of a Morse sequence where,
for $i \in [0,k-1]$, $K_{i+1} \setminus K_i$ is made of a single simplex.
That is, a simplex-wise filtration is a Morse sequence which is made solely of fillings; all faces of $K$ are critical for such a sequence. \\
\end{remark}

\begin{definition} \label{def:seq2}
Let $\ms(K)$ be a Morse sequence. 
The \emph{gradient vector field of 
$\ms$} is the set composed of all regular pairs for $\ms$, that is, the set  $\msc$.
We say that two Morse sequences $\overrightarrow{V}(K)$ and $\ms(K)$
are \emph{equivalent} if
they have the same gradient vector field.\\
\end{definition}

Let $\overrightarrow{V}(K)$ and $\ms(K)$ be two sequences. 
Clearly, if $\overrightarrow{V}(K)$ and $\ms(K)$ are equivalent, then these sequences have the same critical faces. 
In other words, if $\ddot{V} = \msc$, then we have $\widehat{V} = \widehat{W}$. 
The converse of this statement is true if $dim(K) \leq 1$, but not true in general. 
The smallest counter-example is given by an elementary triangle. Let $T$ be the complex composed of all 
non-empty subsets of the set $\{a,b,c\}$. Let $\overrightarrow{V}(T)$ and $\ms(T)$ be the two sequences such that 
$\diamond \overrightarrow{V} = \langle \ka_1, \ldots, \ka_4 \rangle$, $\diamond \ms = \langle \ka'_1, \ldots, \ka'_4 \rangle$ and: 
\begin{enumerate}
\item[-]
$\ka_1 = \{a\}$, $\ka_2= (\{b\}, \{a,b\})$, 
$\ka_3= (\{c\}, \{b,c\})$, $\ka_4= (\{a,c\}, \{a,b,c\})$, 
\item[-] $\ka'_1 = \{a\}$, $\ka'_2= (\{c\}, \{a,c\})$, 
$\ka'_3= (\{b\}, \{b,c\})$, $\ka'_4= (\{a,b\}, \{a,b,c\})$. 
\end{enumerate}
Then we have $\widehat{V} = \widehat{W}= \{ \{a\} \}$, but $\ddot{V} \not= \msc$.

To conclude this section, it is worth mentioning that using Morse sequences as a representation of acyclic vector fields and Morse functions entails no loss of generality. 
Specifically, we can prove that there is an equivalence between the gradient vector field of an arbitrary Morse function and the gradient vector field of a Morse sequence. 
The same holds true for an arbitrary acyclic vector field. For more details on these correspondences, see Appendices \ref{app:vector} and \ref{app:functions}.

In essence, a Morse sequence can be viewed as a total ordering that respects the partial ordering induced by a Morse function or by an acyclic vector field. 
Refer to Definition 5.38 of \cite{Sco19}, which introduces linear extensions of a poset that preserve such partial orderings. 
Additionally, see Theorem 11.9 of \cite{koz20}, which establishes an equivalence between a linear extension and an acyclic vector field. 

\section{The reference and coreference maps}
\label{sec:reference}
\ELIMINE{
A Morse complex is a basic tool for efficiently computing simplicial homology
using discrete Morse theory. Since a Morse complex is built solely on critical complexes, its dimension
is generally much smaller than the one of the original complex.
In this section, we introduce two maps  which will allow us to simplify the construction of such a complex.}

A reference or a coreference map is a linear map which assigns, to each $p$-simplex of a simplicial complex, a set of critical $p$-simplices. 
These maps provide Morse sequences with a structure, they will be the main ingredient of this paper.

Let $\ms(K)$ be a Morse sequence. We write $\ha{W}[p] = \{ c \in K[p] \; | \; c \subseteq \ha{W} \}$.
Thus $\ha{W}[p]$ is a vector subspace of $K[p]$. 
A \emph{frame on $\ms$}
is a map $\Upsilon:$ $\nu \in K \mapsto  \Upsilon(\nu) \subseteq \ha{W}$ such that, for each 
$\nu \in K^{(p)}$, we have $\Upsilon(\nu) \in \ha{W}[p]$.
If $\Upsilon$ is a frame on $\ms$, we write $\Upsilon_p$ for the linear map
$c \in K[p] \mapsto  \Upsilon_p(c) \in \ha{W}[p]$, where $\Upsilon_p(c) = \sum_{\nu \in c} \Upsilon(\nu)$
and $\Upsilon(\emptyset) = 0$.

In the sequel of the paper, we often omit the
subscript or the superscript $p$ if the variable $p$ is clear from the context. \\

\begin{definition} \label{def:reference}
Let $\ms(K)$ be a Morse sequence. Let $\rf$ and $\corf$ be two frames on $\ms$
such that, for each critical simplex $\nu$ of $\ms$,
we have $\rf(\nu) = \corf(\nu) = \nu$.  

\noindent
We say that $\rf$ is a \emph{reference map for $\ms$} if,
for each $\tau \in \ov{W}$, we have: \\
\hspace*{\fill} $\rf(\tau) = 0$ and $ \rf(\partial(\tau)) = 0$. \hspace*{\fill} \\
\noindent
We say that $\corf$ is a \emph{coreference map for $\ms$} if,
for each $\sig \in \un{W}$, we have: \\
\hspace*{\fill} $\corf(\sig) = 0$ and $\corf(\delta(\sig)) = 0$. \hspace*{\fill} \\
\end{definition}

Let $(\sigma,\tau)$ be a regular pair of~$\ms$. By linearity, we observe that: \\
- If $\rf$ is a reference map for $\ms$, then we have $ \rf(\sig) = \rf(\partial(\tau) + \sig)$. \\ 
- If $\corf$ is a coreference map for $\ms$, then have $ \corf(\tau) = \corf(\delta(\sig) + \tau)$. \\ 
That is, we have 
$ \rf(\sig) = \rf(\partial(\tau) \setminus \{ \sig \})$ and $ \corf(\tau) = \corf(\delta(\sig) \setminus \{ \tau \})$. \\
Since a Morse sequence is a filtration, we immediately deduce the following. \\

\begin{proposition} \label{prop:reference0}
Let $\ms(K) = \langle \emptyset = K_0, \ldots, K_k = K \rangle$  be a Morse sequence. 
Let $\rf$ and $\corf$ be, respectively, a reference and a coreference map for $\ms$. Let $\diamond \ms = \langle \ka_1, \ldots, \ka_k \rangle$. 
Then $\rf$ and $\corf$ satisfy the two conditions (A) and (B): 
\begin{enumerate}
\item[(A)]  If $k_i= \nu$ is critical for $\ms$, then we have $\rf(\nu) = \nu$ and $\corf(\nu) = \nu$. 
\item[(B)] If $k_i= (\sig,\tau)$ is a regular pair for $\ms$, then  we have: 
\begin{itemize}
\item
$\rf(\tau) = 0$ and $ \rf(\sig) = \rf(\partial(\tau) + \sig)$, with $\partial(\tau) + \sig \subseteq K_{i-1}$.
\item
$\corf(\sig) = 0$ and $ \corf(\tau) = \corf(\delta(\sig) + \tau)$, with $\delta(\sig) + \tau \subseteq K \setminus K_{i}$.  \\
\end{itemize}
\end{enumerate}
\end{proposition}

If $K = \emptyset$, then the empty map from $\emptyset$ to $\{ \emptyset \}$
is the only reference and also the only coreference map for $\ms$. 

Now, suppose that $K \not= \emptyset$.  Let $\rf$ and $\corf$ be two frames that satisfy the conditions (A) and (B) of Prop.~\ref{prop:reference0}.
It may be seen that we are in the setting of a
recursive characterization of these frames. We have $\rf(\ka_1) = \ka_1$ since $\ka_1$ is critical. We have $\corf(\ka_k) = \ka_k$
if $\ka_k$ is critical. If $\ka_k = (\sig,\tau)$ is a regular pair, we have $\corf(\sig) = 0$, and $\corf(\tau) = \corf(\delta(\sig) + \tau) = \corf(0) = 0$. 
This gives the base cases of the recursive characterization. It follows that $\rf$ and $\corf$
are fully specified. Clearly these frames are reference
and coreference maps. 

Thus, by Prop. \ref{prop:reference0}, we obtain the following result. \\

\begin{theorem} \label{pro:unit}
Any Morse sequence admits a unique reference map and a unique coreference map.\\
\end{theorem}
\noindent


If $\rf$ and $\corf$ are respectively the reference and the coreference maps of $\ms(K)$, we say that 
$(\rf,\corf)$ is \emph{the reference pair of $\ms(K)$}. If $c \in K[p]$, we say that $\rf(c)$
is the \emph{reference of $c$} and $\corf(c)$ is the \emph{coreference of $c$} (for $\ms$). \\
By Prop. \ref{prop:reference0}, there is a simple way to obtain the values of $\rf$ and $\corf$, 
see \cite{BN25}:

\begin{itemize}
\item The reference map $\rf$ can be computed by scanning the sequence $\diamond \ms$
from the left to the right and by applying the conditions (A) and (B) for $\rf$.
\item  The coreference map $\corf$ can be computed
by scanning the sequence $\diamond \ms$
from the right to the left and by applying the conditions (A) and (B) for $\corf$. 
\end{itemize}
In this context, we have an efficient method to label each simplex of a Morse sequence with a set of critical simplices,
providing global information about the topology of the complex.

\ELIMINE{
As a limit case, observe that:
\begin{itemize}[noitemsep,topsep=0pt]
\item If $\tau$ is a facet of $K$, then we have $\Upsilon(\tau) = 0$ whenever $\tau$ is not critical.
\item If $\sig$ is a 0-simplex of $K$, then we have $\Upsilon^*(\sig) = 0$ whenever $\sig$ is not critical.
\end{itemize}
\noindent

 Let $(\rf,\corf)$ be the reference pair of $\ms(K)$ and let $\nu \in K$. We observe that: \\
 - We have $\nu \in \rf (\nu)$ if and only if $\nu$ is critical for $\ms$. \\
 - We have $\nu \in \corf (\nu)$ if and only if $\nu$ is critical for $\ms$. \\
}

\begin{figure*}[tb]
    \centering
    \begin{subfigure}[t]{0.32\textwidth}
        \centering
        \includegraphics[height=.9\textwidth]{Figures/torus.pdf}
        \caption{}
    \end{subfigure}%
    ~
    \begin{subfigure}[t]{0.32\textwidth}
        \centering
        \includegraphics[height=.9\textwidth]{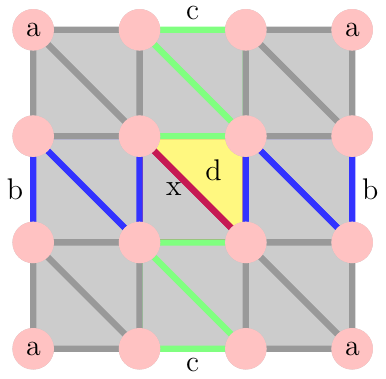}
        \caption{}
    \end{subfigure}%
     ~
    \begin{subfigure}[t]{0.32\textwidth}
        \centering
        \includegraphics[height=.9\textwidth]{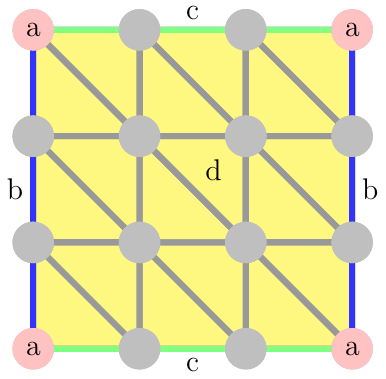}
        \caption{}
    \end{subfigure}
 \caption{The reference map (b) and the coreference map (c) of the Morse sequence of Fig. \ref{fig:MorseSequenceTorus}.
  See text for details.
 }
 \label{fig:MorseRefTorus}
\end{figure*}
We illustrate in Fig.~\ref{fig:MorseRefTorus} (b) the reference map $\rf$ of the Morse sequence of the torus shown in Fig. \ref{fig:MorseSequenceTorus} (f). 
In this figure, any simplex $\nu$ in grey is such that $\rf(\nu)=0$. We have $\rf(\nu)=a$ (in pink) for all simplices $\nu$ of dimension $0$.
All the simplices $\nu$ coloured in blue (resp. in green) are such that $\rf(\nu)=b$ (resp. $\rf(\nu)=c$).
We have $\rf(x) = b+c$ (in red) and $\rf(d) = d$ (in yellow). Using the same colouring  conventions, the coreference map $\corf$ of the same Morse sequence
is given in Fig.~\ref{fig:MorseRefTorus}~(c). 
In Fig.~\ref{fig:MorseRefDunceHat}, again with these conventions, we give the reference and the coreference maps of the Morse sequence 
of the dunce hat depicted in Fig. \ref{fig:MorseSequenceDunceHat} (f). 
In this final example, we observe that each $2$-simplex $\nu$ is such that $\corf(\nu) = c$, where $c$ is the critical $2$-simplex. 
However, this property does not hold if we choose a different critical segment than~$b$ when constructing the Morse sequence. 
In this case we have $\corf(\nu) = 0$ for some $2$-simplices $\nu$.

The following proposition is easy to prove. It provides another example of a reference and a coreference map 
with the simple case where the complex $K$  has a unique critical face, that is, when
$K$ is collapsible. \\

 \begin{proposition} \label{pro:colref} 
 Let $\ms(K)$ be a Morse sequence such that $\ha{W} = \{\ka \}$. Thus, the simplex $\ka$ is a vertex, and $K$ collapses onto $\ka$. Then,  for each $\nu \in K$, 
 \begin{enumerate}
 \item[-] we have $\rf(\nu) = \ka$ if $dim(\nu) = 0$, and $\rf(\nu) = 0$
 if $dim(\nu) \geq 1$, 
  \item[-]  we have $\corf(\nu) = \ka$ if $\nu = \ka$, and $\corf(\nu) = 0$ if $\nu \not= \ka$. \\
  \end{enumerate} 
 \end{proposition}

\begin{figure*}[tb]
    \centering
    \begin{subfigure}[t]{0.32\textwidth}
        \centering
        \includegraphics[width=.9\textwidth]{Figures/duncehat.pdf}
        \caption{}
    \end{subfigure}%
    ~
    \begin{subfigure}[t]{0.32\textwidth}
        \centering
        \includegraphics[width=.9\textwidth]{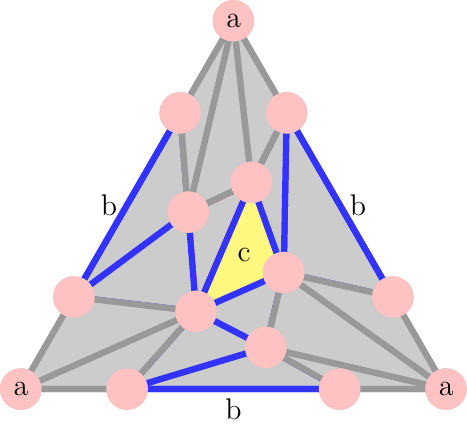}
        \caption{}
    \end{subfigure}%
     ~
    \begin{subfigure}[t]{0.32\textwidth}
        \centering
        \includegraphics[width=.9\textwidth]{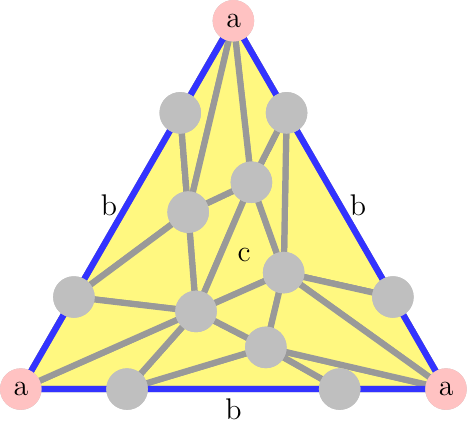}
        \caption{}
    \end{subfigure}
 \caption{
 The reference map (b) and the coreference map (c) of the Morse sequence of Fig. \ref{fig:MorseSequenceDunceHat}.
 See text for details.}
 \label{fig:MorseRefDunceHat}
\end{figure*}


 \begin{proposition} \label{pro:equi} 
 Let $\ms(K)$ and $\ms'(K)$ be two Morse sequences. Let $(\rf,\corf)$ be the reference pair of $\ms$,
 and $(\rf',\corf')$ be the reference pair of $\ms'$. 
If $\ms$ and $\ms'$ are equivalent, then we have
$\rf = \rf'$ and $\corf = \corf'$. \\
 \end{proposition}

\begin{proof}
Let $\ms$ and $\ms'$ be two equivalent sequences. We have $\dd{W} = \dd{W'}$ and $\ha{W} = \ha{W}'$. 
We set $\ms = \langle \emptyset = K_0,\ldots,K_k = K \rangle$ and 
$\diamond \ms = \langle \ka_1, \ldots, \ka_k \rangle$. 
The simplex $\ka_1$ is critical for both $\ms$ and $\ms'$, thus $\rf(\ka_1) =\rf'(\ka_1)$.
Let $i$ such that $2 \leq i \leq k$. Suppose we have $\rf(\nu) = \rf'(\nu)$ for each $\nu \in K_{i-1}$. 
If $\ka_i \in \ha{W}$, then $\ka_i \in \ha{W}'$, thus we have $\rf(\ka_i) =\rf'(\ka_i)= \ka_i$.
If $\ka_i = (\sig,\tau)$ is in $\dd{W}$, then 
$\rf(\tau) = 0$ and $ \rf(\sig) = \rf(\partial(\tau) + \sig)$.
Since $\ka_i$ is in $\dd{W}'$, we have $\rf'(\tau) =0$. 
We have $\partial(\tau) + \sig \subseteq K_{i-1}$, thus by the induction hypothesis 
$ \rf(\sig) = \rf'(\partial(\tau) + \sig) = \rf'(\sig)$. Therefore, we have $\rf(\nu) = \rf'(\nu)$ for each $\nu \in K_{i}$. 
We obtain $\rf = \rf'$. The equality $\corf = \corf'$ may be derived in the same way.
\end{proof}

It should be noted that the converse of the previous proposition is, in general, not true. 
Let us consider again the two sequences $\overrightarrow{V}(T)$ and $\ms(T)$ on the elementary triangle $T$ 
presented in the previous section. 
We have $\widehat{V} = \widehat{W}= \{ \{a\} \}$, but $\ddot{V} \not= \msc$.
By Proposition \ref{pro:colref}, these two sequences have the same reference and coreference maps.

The following result will be used in the sequel. It is simply obtained by linearity from the definitions of $\rf$ and $\corf$. \\

 \begin{proposition} \label{pro:cons1} 

   Let $(\rf,\corf)$ be the reference pair of a Morse sequence $\ms(K)$. 
\begin{enumerate}
\item[-] If $C \subseteq \overline{W}(p)$, then $\rf(C) = 0$ and $\rf(\partial(C)) = 0$.   
\item[-] If $D \subseteq \underline{W}(p)$, then $\corf(D) = 0$ and $\corf(\delta(D)) = 0$.
\end{enumerate}
 \end{proposition}



\section{Gradient paths and cogradient paths}
\label{sec:grad}
Reference maps are closely related to the notion of a gradient path.
We first recall the classical definition of such a path 
(see also \cite{forman2002discretecohomology}).\\
Let $\ms(K)$ be a Morse sequence. We consider the sequences: 
 \begin{itemize}
\item $\pi = \langle \sigma_0, \tau_0, \ldots,\sigma_{k-1}, \tau_{k-1}, \sig_k \rangle$, $k \geq 0$,
 with $\sigma_i \in K^{(p)}$, $\tau_i \in K^{(p+1)}$. 
We say that $\pi$ is a \emph{gradient path in $\ms$ (from $\sig_0$ to $\sig_k$)} if, for any $i \in [0,k-1]$,
the pair $(\sigma_i,\tau_{i})$ is regular for $\ms$ and
$\sigma_{i+1} \in \partial(\tau_{i})$, with $\sigma_{i+1} \not= \sig_{i}$.
The path $\pi$ is {\em trivial} if $k=0$, that is, if
$\pi = \langle \sigma_0 \rangle$ with $\sigma_0 \in K^{(p)}$.
\item $\pi = \langle \tau_0, \sig_1,\tau_1, \ldots,\sig_{k}, \tau_{k} \rangle$, $k \geq 0$, with $\tau_i \in K^{(p)}$, $\sigma_i \in K^{(p-1)}$.
We say that $\pi$ is a \emph{cogradient path in $\ms$ (from $\tau_0$ to $\tau_k$)} if, for any $i \in [1,k]$,
the pair $(\sigma_i,\tau_{i})$ is regular for $\ms$ and
$\tau_{i-1} \in \delta(\sig_{i})$, with $\tau_{i} \not= \tau_{i-1}$. 
The path $\pi$ is {\em trivial} if $k=0$, that is, if 
$\pi = \langle \tau_0 \rangle$ with $\tau_0 \in K^{(p)}$. 
\ELIMINE{
\item Let $\pi = \langle \tau_0, \sig_0, \ldots, \tau_{k-1}, \sig_{k-1}, \tau_{k} \rangle$, $k \geq 0$,
be a sequence with $\tau_i \in K^{(p)}$, $\sigma_i \in K^{(p-1)}$.
We say that $\pi$ is a \emph{cogradient path in $\ms$ (from $\tau_0$ to $\tau_k$)} if, for any $i \in [0,k-1]$,
the pair $(\sigma_i,\tau_{i})$ is regular for $\ms$ and
$\tau_{i+1} \in \delta(\sig_{i})$, with $\tau_{i+1} \not= \tau_{i}$.
The path $\pi$ is {\em trivial} if $k=0$, that is, if
$\pi = \langle \tau_0 \rangle$ with $\tau_0 \in K^{(p)}$. 
}
\end{itemize}
Observe that a gradient path may end at a critical face, and a cogradient path may begin at a critical face. 
Also note that the sequence obtained by reversing a gradient path is not a cogradient path unless the path is trivial.

 A gradient vector field which may be obtained from a Morse sequence is given in Figure \ref{fig:Grad} (a).
The 1D (resp. 2D) regular pairs correspond to black (resp. green) arrows, the critical faces are in red. 
There are two gradient paths from the face $a$ to the critical face $b$. 
The gradient path beginning at $a$ splits into two gradient paths at a triangle, and these two paths merge at a segment. 
Note that this gradient vector field is not obtained by a maximal increasing or decreasing scheme, 
otherwise we would have only one critical face.

\begin{figure*}[tb]
    \centering
    \begin{subfigure}[t]{0.5\textwidth}
        \centering
        \includegraphics[width=0.9\textwidth]{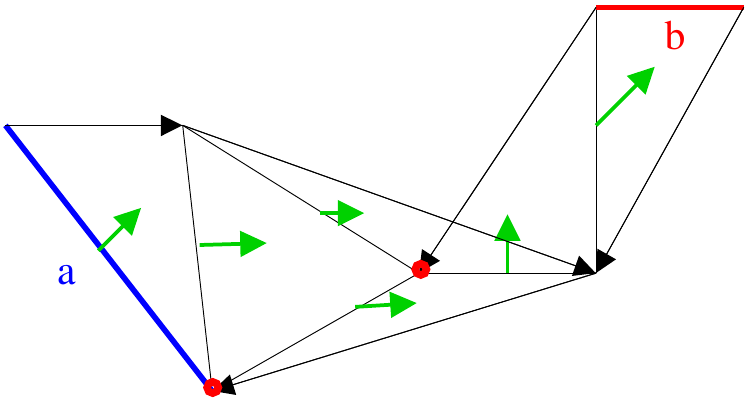}
        \caption{}
    \end{subfigure}%
    ~
    \begin{subfigure}[t]{0.5\textwidth}
        \centering
        \includegraphics[width=0.9\textwidth]{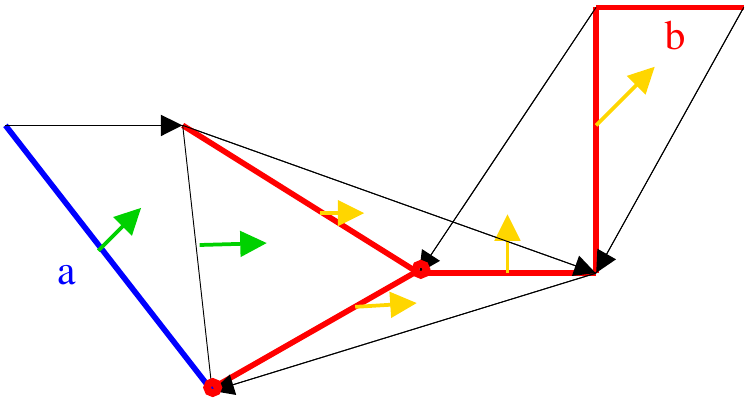}
        \caption{}
    \end{subfigure}%
 \caption{(a) A gradient vector field: there are two gradient paths from $a$ to $b$. \\
  (b)  The only regular pairs involved in a $\rf$-path ending at $b$ are in yellow: thus there is no $\rf$-path from $a$ to $b$. See text for details.}
 \label{fig:Grad}
\end{figure*}

The parity of the number of gradient paths between two simplices  is a classical tool for extracting information about the homology of a complex. 
With our definitions, this may be accomplished through the following result. \\

 \begin{theorem} \label{pro:grad1} 
 Let $(\rf,\corf)$ be the reference pair of a Morse sequence $\ms(K)$. \\
 If $\kappa \in \ha{W}$ and $\nu \in K$, then: 
 \begin{enumerate}
\item $\kappa \in \rf (\nu)$ if and only if the number of gradient paths from $\nu$ to $\kappa$ is odd. 
\item $\kappa \in \corf (\nu)$ if and only if the number of cogradient paths from  $\kappa$ to $\nu$ is odd. \\
\end{enumerate}
 \end{theorem}

\begin{proof}  
Let $\ms (K) = \langle \emptyset = K_0,\ldots, K_i,\ldots, K_k =K \rangle$. \\
1) For the first statement we 
proceed by induction on $i$, the base case $\nu \in K_0$ is trivial. 
Suppose the property is true for each $\nu \in K_{i-1}$, $i \geq 1$.
We write $N(\nu,\kappa)$ for the number of gradient paths from $\nu$ to $\kappa$.
Let $\nu \in K_i \setminus K_{i-1}$. 

i)  Suppose $\nu \in \ha{W}$. We have $\rf(\nu) = \nu$. If $\nu \not= \kappa$ there is no gradient path from $\nu$ to $\kappa$, thus $N(\nu,\kappa) = 0$. 
If $\nu = \kappa$ there is one and only one gradient path from $\nu$ to $\kappa$, that is, the trivial path 
$\langle \kappa \rangle$, thus $N(\nu,\kappa) = 1$. 

ii) If $\nu \not\in \ha{W}$, we have $K_i \setminus K_{i-1}= \{ \sig,\tau \}$ where $(\sig,\tau)$ is a regular pair for $\ms$. \\
- Suppose $\nu = \tau$. We have $\rf (\tau) =0$ and $\tau \not= \ka$. Furthermore, there is no non-trivial gradient path starting at $\tau$.
Thus we have $N(\nu,\kappa) = 0$. \\
- Suppose $\nu = \sig$.  We have $\sig \not= \ka$.
We consider the sets: \\
\hspace*{\fill}
 $A = \{ \nu' \in \sig + \partial{\tau} \}$, $B = \{ \nu' \in A \; | \; N(\nu',\kappa)$ is odd$\}$, 
$C = \{ \nu' \in A \; | \; \kappa \in \rf (\nu') \}$. \hspace*{\fill} \\
 Any gradient path from $\nu$ to $\kappa$ is obtained by concatenating $(\sig,\tau)$ to a gradient path from a simplex $\nu' \in A$ to $\kappa$. Thus we have
$N(\nu,\kappa) = \sum_{\nu' \in A} N(\nu',\kappa)$. By parity considerations, we obtain: \\
\hspace*{\fill}
 $N(\nu,\kappa)$ is odd  if and only if $Card(B)$ is odd.
  \hspace*{\fill} $(P)$ \\
By the definition of $\rf$ we have: \\
\hspace*{\fill}
 $\kappa \in \rf(\nu)$  if and only if $Card(C)$ is odd.
  \hspace*{\fill} $(Q)$ \\
 We have $A \subseteq K_{i-1}$, thus $B \subseteq K_{i-1}$ and $C \subseteq K_{i-1}$. By our induction hypothesis we obtain $B = C$. 
The result follows from  $(P)$ and $(Q)$. \\ 
2) For the second statement we consider the sequence \\
\hspace*{\fill}
$\overleftarrow{W} (K) = \langle \overline{K}_0 = K \setminus K_k,\ldots, \overline{K}_j = K \setminus K_{k-j},\ldots, \overline{K}_k =  K \setminus K_{0} \rangle$. 
\hspace*{\fill} \\
We proceed by induction on the number $j$, the base case $\nu \in \overline{K}_0$ is trivial.
Then we consider $\nu \in \overline{K}_j \setminus \overline{K}_{j-1}$. 
We use the same kind of arguments as above. 
The set $A$ is replaced by the set $A' = \{ \nu' \in \tau + \delta{\sig} \}$,
thus $A' \subseteq  \overline{K}_{j-1}$. 
\end{proof}

Now, we introduce a refinement of gradient and cogradient paths. \\

\begin{definition} \label{rfpaths}
Let $(\rf,\corf)$ be the reference pair of a Morse sequence $\ms$. \\
We say that a gradient path $\langle \sigma_0, \tau_0, \ldots,\sigma_{k-1}, \tau_{k-1}, \sig_k \rangle$ in $\ms$ is a \emph{$\rf$-path} if,
for each $i \in [0,k]$, we have $\sig_k \in \rf(\sig_i)$. \\
We say that a cogradient path $\langle \tau_0, \sig_1,\tau_1, \ldots,\sig_{k}, \tau_{k} \rangle$
in $\ms$ is a \emph{$\corf$-path} if, for each $i \in [0,k]$, we have $\tau_0 \in \corf(\tau_i)$. \\
\end{definition}

We observe that a $\rf$-path  (resp. a $\corf$-path) necessarily ends (resp. begins) at a critical face.
By using the same arguments as in the proof of Theorem \ref{pro:grad1}, we derive:\\

 \begin{proposition} \label{pro:grad2}
 Let $(\rf,\corf)$ be the reference pair of a Morse sequence $\ms(K)$. \\
If $\kappa \in \ha{W}$ and $\nu \in K$, then: 
 \begin{enumerate}
\item $\kappa \in \rf (\nu)$ if and only if the number of $\rf$-paths from $\nu$ to $\kappa$ is odd. 
\item $\kappa \in \corf (\nu)$ if and only if the number of $\corf$-paths from  $\kappa$ to $\nu$ is odd.\\
 \end{enumerate}
 \end{proposition}
 
Therefore, we can deduce that: \\

 \begin{corollary} \label{cor:grad3}
 Let $(\rf,\corf)$ be the reference pair of a Morse sequence $\ms(K)$. \\
 If $\kappa \in \ha{W}$ and $\nu \in K$, then:
  \begin{enumerate}
 \item We have $\kappa \in \rf (\nu)$ if and only if there exists a $\rf$-path from $\nu$ to $\kappa$. 
\item We have $\kappa \in \corf (\nu)$ if and only if there exists a $\corf$-path from $\kappa$  to $\nu$.\\
 \end{enumerate}
 \end{corollary}

In Figure \ref{fig:Grad} (b), the only regular pairs involved in a $\rf$-path ending at the critical face $b$ are in yellow: 
we can check that there are precisely five segments $\nu$ such that $b \in \rf (\nu)$, these segments are in red.
We can also check that there is no $\rf$-path from $a$ to $b$. By the above corollary, it means that $b \not\in \rf (a)$. 
Since there are precisely two gradient paths from $a$ to $b$, this fact is also a consequence of Theorem \ref{pro:grad1}.

Thus, the property given in Corollary \ref{cor:grad3} does not hold if we consider arbitrary gradient (or cogradient) paths instead of $\rf$-paths (or $\corf$-paths):
$\rf$- and $\corf$-paths allow us to discard some paths that do not convey information about homology.


\section{The critical and cocritical complexes}
\label{sec:Mcomplex}
Now, thanks to the reference and the coreference of a Morse sequence $\ms$,
we introduce two maps  $\plc_p$ and $\dlc_p$ that are restricted to $\ha{W}$.
With these maps we obtain the critical and the cocritical complexes of a Morse sequence. \\
In the sequel of the paper, if we write $\ms$ for a Morse sequence, the symbols $\rf$ and $\corf$
stand respectively for the reference and coreference maps of $\ms$. \\

\begin{definition} \label{def:criticbound}
Let $\ms$ be a Morse sequence. Let $\plc_p$ and $\dlc_p$ be the two linear maps: \\
\hspace*{\fill}
$\plc_p:$ $\ha{W}[p] \rightarrow  \ha{W}[p-1]$ and $\dlc_p:$ $\ha{W}[p] \rightarrow  \ha{W}[p+1]$,
\hspace*{\fill} \\
such that, 
for each $\nu \in \ha{W}(p)$, 
$\plc(\nu) = \rf_{p-1}(\pl(\nu))$ and  $\dlc(\nu) = \corf_{p+1}(\dl(\nu))$. \\
If $c \in  \ha{W}[p]$, then $\plc_p(c)$ is the \emph{$\rf$-boundary of $c$}, $\dlc_p(c)$
is the \emph{$\corf$-coboundary of $c$}.\\
\end{definition}

Thus, for each $c \in  \ha{W}[p]$, we have: \\
- $\plc_p(c)  = \sum_{\nu \in c} \rf_{p-1}(\partial(\nu)) = \rf_{p-1} (\sum_{\nu \in c} \partial(\nu)) = \rf_{p-1}(\pl_p(c))$, and \\
- $\dlc_p(c) = \sum_{\nu \in c} \corf_{p+1}(\delta(\nu)) = \corf_{p+1} (\sum_{\nu \in c} \delta(\nu)) = \corf_{p+1}(\dl_p(c))$.

As an illustration, let us consider the following case of the dunce hat: \\
- In Fig.~\ref{fig:MorseRefDunceHat} (b), we see that $\plc(c) = \rf_1(\partial(c))=b+b+b=b$.\\
- In Fig.~\ref{fig:MorseRefDunceHat} (c), we see that $\dlc(b) = \corf_2(\delta(b))=c+c+c=c$.

This example suggests us the following theorem, which reflects an important duality relation between the reference and the coreference maps of a Morse sequence. \\

  \begin{theorem} \label{pro:grad11} Let $\ms$ be a Morse sequence on $K$
 and let $\sig, \tau \in \ha{W}$.
 We have: \\
 \hspace*{\fill}
  $\sig \in \plc(\tau)$ if and only if $\tau \in \dlc(\sig)$.
   \hspace*{\fill}
 \end{theorem}

\ELIMINE{
\begin{proof} Let $\sig, \tau \in \msc$. Suppose $\sig \in \plc(\tau)$, thus
$\sig \in  \rf_{p-1}(\pl(\tau))$.
By Prop. there exists a $\rf$-path $\langle \sig' = \sigma_0, \tau_0, \ldots,\sigma_{k-1}, \tau_{k-1}, \sig_k = \sig \rangle$ 
from  a simplex $\sig' \in \pl(\tau)$ to~$\sig$. Thus, for each $i \in [0,k]$ , we have $\sig \in \rf(\sig_i)$.
Now let us consider cogradient path $\langle \tau_{k-1}, \sigma_{k-1}, \ldots, \tau_{0},  \sigma_0 = \sig' ,\tau \rangle$,

in $\ms$ is a \emph{$\corf$-path} if, for each $i \in [0,k]$, we have $\tau_k \in \corf(\tau_i)$.
 \qed
\end{proof}
}

\begin{proof} Let $\sig, \tau \in \ha{W}$. We write: \\
- $\Omega_\sig$ for the set of all gradient paths from a simplex $\nu \in \partial(\tau)$
to $\sig$, \\
- $\Omega^\tau$ for the set of all cogradient paths from $\tau$ to a simplex $\nu \in \delta(\sig)$. \\
We have $\sig \in \plc(\tau)$ if and only if $Card( \{ \nu \; | \; \nu \in \partial(\tau)$ and $\sig \in \rf(\nu) \})$ is odd. \\
Thus, by Theorem \ref{pro:grad1} and by parity considerations, we have $\sig \in \plc(\tau)$ if and only if $Card(\Omega_\sig)$ is odd. 
Similarly we have $\tau \in \dlc(\sig)$ if and only if $Card(\Omega^\tau)$ is odd. \\
We will establish the result by showing that $Card(\Omega_\sig) = Card(\Omega^\tau)$. \\
1) Let $\pi$ be a gradient path in $\Omega_\sig$. Then $\pi$ is of the form \\
\hspace*{\fill}
$\pi = \langle \sigma_0, \tau_0, \ldots,\sigma_{k-1}, \tau_{k-1}, \sig_k \rangle$, \hspace*{\fill} \\
with $\sig_0 \in \partial(\tau)$ and $\sigma_k = \sig$.
Let
$f(\pi)$ be the sequence \\
\hspace*{\fill}
$f(\pi) = \langle \tau, \sigma_0, \tau_0, \ldots,\sigma_{k-1}, \tau_{k-1} \rangle$. \hspace*{\fill} \\
We can check that $f(\pi)$ is a cogradient path from $\tau$ to $\tau_{k-1}$.
Furthermore we have $\tau_{k-1} \in \delta(\sig)$, thus $f(\pi)$ is in $\Omega^\tau$. \\
Let $f$ be the map $\pi \in \Omega_\sig \mapsto f(\pi) \in \Omega^\tau$ and let \\
\hspace*{\fill}
$\pi = \langle \sigma_0, \tau_0, \ldots,\sigma_{k-1}, \tau_{k-1}, \sig_k \rangle$
and $\pi' = \langle \sigma'_0, \tau'_0, \ldots,\sigma'_{l-1}, \tau'_{l-1}, \sig'_l \rangle$
\hspace*{\fill} \\
be two distinct paths in $\Omega_\sig$. Since $\sig = \sig_k = \sig'_l$, the two sequences \\
\hspace*{\fill}
$\langle \sigma_0, \tau_0, \ldots,\sigma_{k-1}, \tau_{k-1} \rangle$ and
$\langle \sigma'_0, \tau'_0, \ldots,\sigma'_{l-1}, \tau'_{l-1} \rangle$ \hspace*{\fill} \\
must be distinct.
It follows that $f(\pi) \not= f(\pi')$. Therefore the map $f$ is injective, which means that
$Card(\Omega_\sig) \leq Card(\Omega^\tau)$. \\
2) Now, let $\pi = \langle \tau_0, \sigma_1, \tau_1, \ldots,\sigma_{k}, \tau_{k} \rangle$ be a path in $\Omega^\tau$.
Thus, we have $\tau_0 = \tau$ and $\tau_k \in \delta(\sig)$. We proceed as above: \\
- The sequence $g(\pi) = \langle \sigma_1, \tau_1, \ldots,\sigma_{k}, \tau_{k}, \sig \rangle$
is a gradient path in $\Omega_\sig$, \\
- The map $g : \pi \in \Omega^\tau \mapsto g(\pi) \in \Omega_\sig$ is injective, \\
- Therefore we have $Card(\Omega^\tau) \leq Card(\Omega_\sig)$. 
\end{proof}


 \begin{theorem} \label{pro:label3}
  Let  $\ms$ be a Morse sequence. We have:\\
 \hspace*{\fill}
  $\plc_p \circ \rf_p = \rf_{p-1} \circ \partial_p$ and
  $\dlc_p \circ \corf_p = \corf_{p+1} \circ \delta_p$.
   \hspace*{\fill} \\
 \end{theorem}

We will only give the first part of the proof of the theorem. 
The second part may be derived by simply exchanging the role of the operators 
$\partial$ and $\delta$. 

\emph{In the sequel
of the paper, we will proceed in the same manner for all propositions and theorems where we have such a duality. }

\begin{proof}
Let $\ms = \langle \emptyset = K_0,...,K_k =K \rangle$ be a Morse sequence on $K$, and let 
$\diamond \ms = \langle \ka_1, \ldots, \ka_k \rangle$. 
We consider the statement $(S_i)$: For each $c \in K_i[p]$,
 we have $\plc_p (\rf_{p} (c)) = \rf_{p-1} (\partial_p (c))$.
 We have $K_0[p] = \{\emptyset \}$. Thus $(S_0)$ holds. \\
 Suppose $(S_{i-1})$ holds, with $0 \leq i-1 \leq k-1$, and
let $c \in K_i[p]$. \\
1) Suppose $\ka_i = \nu$, with $\nu \in \ha{W}$. If $\nu \not\in c$, then we are done.
Otherwise, we have $c = c' \cup \{ \nu \}$, with $c' \in K_{i-1}[p]$.
Thus $\partial_p(c) = \partial_p (c') + \partial(\nu)$ and
$\rf_{p-1} (\partial_p (c)) = \rf_{p-1} (\partial_p (c')) + \rf_{p-1} (\partial(\nu))$.
By the induction hypothesis and by the definition of $\plc(\nu)$, we obtain
$\rf_{p-1} (\partial_p(c)) = \plc_p (\rf_{p} (c'))  + \plc(\nu)$.
Since the reference map $\rf$ is the identity on critical faces, we may write:\\
\hspace*{\fill}
$\rf_{p-1} (\partial_p(c)) = \plc_p (\rf_{p} (c')) + \plc_p(\rf_{p} (\nu ))$. \hspace*{\fill} \\
By linearity we obtain $\rf_{p-1} (\partial_p(c)) =  \plc_p (\rf_{p} (c))$. \\
2) Suppose $\ka_i  = (\sigma,\tau)$ is a free pair.
If $\sigma \not\in c$ and $\tau \not\in c$, then we are done. \\
2.1) Suppose $\sigma \in c$. Let $c' = c + \partial_{p+1}(\tau)$.
By Proposition \ref{pro:cons1}, we have $\rf_{p} (c') = \rf_{p} (c) + \rf_{p} (\partial_{p+1}(\tau)) = \rf_{p} (c)$. We also have
$\partial_p(c') = \partial_p(c) + \partial_p(\partial_{p+1}(\tau)) = \partial_p(c)$.  \\
But we see that $\sig \not\in c'$.
By the induction hypothesis, it follows that
$\plc_p (\rf_{p} (c')) = \rf_{p-1} (\partial_p (c'))$. By the previous equalities, we obtain
$\plc_p (\rf_{p} (c)) = \rf_{p-1} (\partial_p (c))$.\\
2.2) Suppose $\tau \in c$. We have $c = c' \cup \{ \tau \}$, with $c' \in K_{i-1}[p]$. Since $\rf (\tau) =0$, we obtain
$\rf_{p} (c) = \rf_{p} (c')$. Furthermore $\rf_{p-1} (\partial_p (c)) = \rf_{p-1} (\partial_p (c')) + \rf_{p-1} (\partial_p (\tau))
= \rf_{p-1} (\partial_p (c'))$.
By the induction hypothesis, we have $\plc_p (\rf_{p} (c')) = \rf_{p-1} (\partial_p (c'))$.
Therefore $\plc_p (\rf_{p} (c)) = \rf_{p-1} (\partial_p (c))$. \\
We conclude that, in all cases, the statement $(S_i)$ holds. 
\end{proof}

The two following results are direct consequences of Theorem \ref{pro:label3}.\\

 \begin{proposition} \label{pro:label1}
    Let $\ms(K)$ be a Morse sequence and let $c,c' \in K[p]$. 
    \begin{enumerate}
\item We have
$\rf_{p-1}(\partial_p(c)) = \rf_{p-1}(\partial_p(c'))$ whenever $\rf_{p} (c) = \rf_{p} (c')$.
\item We have
$\corf_{p+1}(\delta_p(c)) = \corf_{p+1}(\delta_p(c'))$ whenever $\corf_{p} (c) = \corf_{p} (c')$.
\end{enumerate}
 \end{proposition}

\begin{proof}
Let $c,c' \in K[p]$ with $\rf_{p} (c) = \rf_{p} (c')$. Thus $\plc_p  (\rf_{p} (c)) = \plc_p  (\rf_{p} (c'))$.
By Theorem \ref{pro:label3}, we have $\rf_{p-1} (\partial_p (c)) = \rf_{p-1} (\partial_p (c'))$. 
\end{proof}

 \begin{proposition} \label{pro:label4}
If $\ms$ is a Morse sequence, then: 
\begin{enumerate}
\item The maps $\plc_p $ are boundary operators. That is,
we have $\plc_p  \circ \plc_{p+1} = 0$. 
\item The maps $\dlc_p $ are coboundary operators. That is,
we have $\dlc_{p+1}  \circ \dlc_{p} = 0$.
\end{enumerate}
\end{proposition}

\begin{proof}
Let $\nu \in \ha{W}$, with $\nu \in K^{(p+1)}$.
We have $\plc_{p+1} (\nu) = \rf_{p} (\partial_{p+1} (\nu))$.
By Theorem \ref{pro:label3}, we have
$\plc_p(\rf_{p} (\partial_{p+1} (\nu))) = \rf_{p-1} (\partial_p ( \partial_{p+1} (\nu))) = \rf_{p-1} (0) = 0$. \\
Thus $\plc_p (\plc_{p+1} (\nu)) = 0$, which gives the result by linearity. 
\end{proof}


Since  $\plc_p \circ \plc_{p+1} = 0$, the pair $( \ha{W}[p], \plc_p )$ satisfies the definition of a chain complex, see Section \ref{sec:basic}.
Similarly, $(\ha{W}[p], \dlc_p)$ satisfies the definition of a cochain complex. \\

\begin{definition}
Let $\ms$ be a Morse sequence. 
We say that the chain complex $( \ha{W}[p], \plc_p )$ is the \emph{critical complex of $\ms$}
and the cochain complex $( \ha{W}[p], \dlc_p)$ is the \emph{cocritical complex of $\ms$}. \\
\end{definition}

This notion of a critical complex
is equivalent to the classical notion, also sometimes called \emph{Morse complex}, given in the context of discrete Morse theory.
This fact may be verified using  Theorem 8.31 of \cite{Sco19}, Theorem \ref{pro:grad1}, and the very definition of the differential $\plc_p$.


As an example, we give hereafter the complete description of the critical and the cocritical complex of the Morse sequence of the torus depicted 
Fig.~\ref{fig:MorseSequenceTorus}. 
For this sequence $\ms$, we have $\ha{W}^{(0)} = \{a\}$, $\ha{W}^{(1)} = \{b,c\}$, and $\ha{W}^{(2)} = \{d\}$. 
The values of the maps 
$\plc_p$ and  $\dlc_p$ may be obtained from  the figures ~\ref{fig:MorseRefTorus} (b) and ~\ref{fig:MorseRefTorus} (c). We have: \\ 
- $\plc_0(a) = 0$, $\plc_1(b) = \rf_0(\pl(b))= a + a = 0$, $\plc_1(c) = \rf_0(\pl(c))= a + a = 0$, and 
$\plc_2(d) = \rf_1(\pl(d)) = b + c + \rf(x) = 0$. \\
- $\dlc_0(a) = \corf_1(\delta(a)) = b+b+c+c+0+0 = 0$, $\dlc_1(b) = \corf_2(\delta(b)) = d + d = 0$, $\dlc_1(c) = \corf_2(\delta(c)) = d + d = 0$,
$\dlc_2(d) = \corf_3(\delta(d)) = 0$. \\
Thus, for each $p$, we have $\plc_p=0$ and  $\dlc_p=0$. It means that, for each $p$, $\ha{W}^{(p)}$ is a basis for $H_p(\ha{W})$ and for $H^p(\ha{W})$.


By Theorem \ref{pro:label3}, the map $\rf$ is a \emph{chain map} \cite{Hat01} from the chain complex $( K[p], \partial_p )$
to  the chain complex $( \ha{W}[p], \plc_p )$.
The map $\corf$ is a \emph{cochain map} \cite{Hat01} from the chain complex $( K[p], \delta_p )$
to  the chain complex $(\ha{W}[p], \dlc_p)$. 

Also, we easily derive from Theorem \ref{pro:label3} the following facts, which are basic properties of chain maps. We follow the notations given 
in subsection \ref{subsec:chain}: \\
- If $c \in Z_p(K)$, then $\rf_{p} (c) \in Z_p(\ha{W})$, if $c \in Z^p(K)$, then $\corf_{p} (c) \in Z^p(\ha{W})$. \\
- If $c \in B_p(K)$, then $\rf_{p} (c) \in B_p(\ha{W})$, if $c \in B^p(K)$, then $\corf_{p} (c) \in B^p(\ha{W})$. \\
It follows that we have linear maps (homomorphisms) between $H_p(K)$ and $H_p(\ha{W})$,
and between $H^p(K)$ and $H^p(\ha{W})$: \\
\hspace*{\fill}
$\rf^H_p: [z] \in H_p(K) \mapsto \rf^H_p ([z]) = [\rf_{p} (z)] \in H_p(\ha{W})$, 
\hspace*{\fill} \\
\hspace*{\fill}
$\corf^H_p: [z] \in H^p(K) \mapsto \corf^H_p ([z]) = [\corf_{p} (z)] \in H^p(\ha{W})$. 
\hspace*{\fill} 



Thus, the maps $\rf$ and $\corf$ allow to carry out the homology of the complex $K$ to the smaller space $\ha{W}$. We will
see in section \ref{sec:ext} how to go in the other direction in order to obtain an isomorphism between the vector spaces 
$H^p(K)$ and $H^p(\ha{W})$. 

\section{Arranged Morse sequences and skeletons}
\label{sec:reg}

In this section, we introduce \emph{arranged sequences}, which can be obtained by reordering Morse sequences. 
We also define the related notions of \emph{lower} and \emph{upper skeletons} for an arbitrary Morse sequence. Notably, the properties of skeletons and arranged sequences play a crucial role in the theorems presented in the following sections.

Recall that the dimension of a free pair $(\sig,\tau)$ is equal to $dim(\tau)$. 
An arranged sequence is simply a Morse sequence where the critical faces and the free pairs are partially ordered according to their dimensions
\footnote{See also the rephrasing of the proof of Proposition 13.1 given in \cite{koz20} where a similar ordering on the simplices is used.}. \\

\begin{definition} \label{def:reg0}
Let $\ms$ be a Morse sequence and let $\diamond \ms = \langle \ka_1, \ldots, \ka_k \rangle$. 
We say that $\ms$ is an \emph{arranged (Morse) sequence} if for each $i \in [1,k-1 ]$, we have: 
\begin{enumerate}
\item $dim(\ka_{i}) \leq dim(\ka_{i+1})$, and 
\item $dim(\ka_{i}) < dim(\ka_{i+1})$ if $\ka_i$ is critical and $\ka_{i+1}$ is a regular pair. \\
\end{enumerate}
\end{definition}

\noindent
Let us write
 $\ha{W}^{(p)} = \{ \nu \in \ha{W} \; | \; dim(\nu) = p \}$
and $\dd{W}^{(p)} = \{ (\sig,\tau) \in \dd{W} \; | \; dim(\tau) = p \}$.
If $\ms$ is an arranged sequence, each of the two sets $\ha{W}^{(p)}$ and $\dd{W}^{(p)}$
constitutes  a substring of $\diamond \ms$, that is a contiguous sequence. Furthermore, $\ha{W}^{(p)} \cup \dd{W}^{(p)}$
constitutes also a substring of $\diamond \ms$ where $\dd{W}^{(p)}$ is a prefix. 

Let $\diamond \ms = \langle \ka_1, \ldots, \ka_k \rangle$ be a simplex-wise Morse sequence.
If $\ka_i$ and $\ka_{i+1}$ do not satisfy one of the two above conditions, 
it can be checked that the sequence obtained by swapping $\ka_i$ and 
$\ka_{i+1}$ is still a simplex-wise Morse sequence. By induction, we have the following result. \\

 \begin{theorem} \label{pro:reg1}
If $\ms$ is a Morse sequence, then there exists an arranged sequence $\overrightarrow{V}$ such that $\ms$ and $\overrightarrow{V}$
are  equivalent.\\
 \end{theorem}

\begin{remark} \label{rem:arr}
It should be noted that the maximal increasing scheme introduced in Section \ref{sec:seq} does not, in general, directly produce an arranged Morse sequence. 
A sorting process is required to obtain an equivalent arranged sequence. The same applies to the maximal decreasing scheme. 
For instance, it can be seen that the two Morse sequences of the torus and the dunce hat, shown in Fig. \ref{fig:MorseSequenceTorus} and Fig. \ref{fig:MorseSequenceDunceHat},
do not satisfy the second condition of Definition \ref{def:reg0}. 
In fact, such schemes are not required to fulfill this condition, as the goal is to minimize the number of critical faces in the sequence. \\
\end{remark}

Recall that the $p$-skeleton of a simplicial complex $K$ is the set composed of all the $q$-simplices of $K$ such that $q \leq p$. 
The following may be seen as a refinement of this notion for Morse sequences.\\
 
\begin{definition} \label{def:reg1}
Let $\ms$ be a Morse sequence. We write: 
\begin{itemize}
\item $W_p^- = \{ \nu \in  \ov{W} \; | \; dim(\nu) \leq p \} 
\cup \{ \nu \in \ha{W} \cup \un{W} \; | \; dim(\nu) \leq p-1 \}$. 
\item $W_p^+ = \{ \nu \in \ha{W} \cup \ov{W} \; | \; dim(\nu) \leq p \} 
\cup \{ \nu \in \un{W} \; | \; dim(\nu) \leq p-1 \}$. 
\end{itemize}
The sets $W_p^-$ and $W_p^+$ are, respectively, the \emph{lower} and \emph{the upper $p$-skeleton of~$\ms$}. 
 Also, if $d = dim(K)$, we say that the sequence: \\
 \hspace*{\fill}
  $\overrightarrow{S}(W) = \langle W_0^-,W_0^+, \ldots, W^-_p,W_p^+, \ldots,  W_d^-, W_d^+ \rangle$
   \hspace*{\fill} \\
   is the \emph{skeleton sequence of $\ms$}. \\
\end{definition}

Thus, $W_p^-$ is the simplicial complex that contains the $(p-1)$-skeleton of $K$ and the $p$-simplices 
that are upper regular for $\ms$; $W_p^+$ is the simplicial complex that contains the $p$-skeleton of $K$
except the $p$-simplices of $K$ that are  lower regular for $\ms$. We observe that: 
\begin{enumerate}
\item[-] We have $W_0^- = \emptyset$ and $W_d^+ = K$. The set $W_0^+$ is made of all vertices of $K$ that are critical for $\ms$. 
\item[-] The skeleton sequence $\overrightarrow{S}(W)$ of $\ms$ is a filtration (see Remark 3). 
\item[-]  If two Morse sequences $\overrightarrow{V}$ and $\ms$ are equivalent, then $\overrightarrow{S}(V) = \overrightarrow{S}(W)$. 
\item[-] We have $W_{p}^+ \setminus W_{p}^- = \ha{W}^{(p)}$. 
\item[-] We have $W_{p+1}^- \setminus W_{p}^+ = \{ \nu \in  \ov{W} \; | \; dim(\nu) = p+1 \} 
\cup \{ \nu \in \un{W} \; | \; dim(\nu) = p \}$. 
That is, $W_{p+1}^- \setminus W_{p}^+$ is made of all simplices that are in $\dd{W}^{(p)}$. 
\end{enumerate}
\begin{figure*}[tb]
    \centering
    \begin{subfigure}[t]{0.32\textwidth}
        \centering
        \includegraphics[height=.9\textwidth]{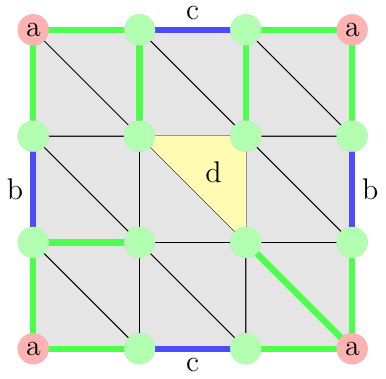}
        \caption{}
    \end{subfigure}%
    ~ \hspace*{1.5cm}
    \begin{subfigure}[t]{0.32\textwidth}
        \centering
        \includegraphics[height=.9\textwidth]{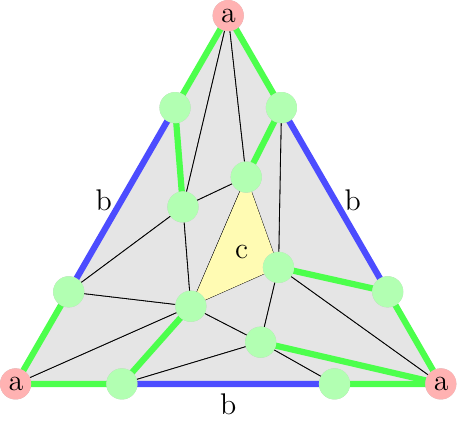}
        \caption{}
    \end{subfigure}%
 \caption{Skeleton sequences of two Morse sequences $\ms(K)$. In (a) the sequence $\ms(K)$ is the one corresponding to Fig \ref{fig:MorseSequenceTorus}, 
 and in (b) the sequence $\ms(K)$ corresponds to Fig~\ref{fig:MorseSequenceDunceHat}. 
  See text for details.
 }
 \label{fig:MorseSequenceTorusSkeleton}
\end{figure*}
An illustration of two skeleton sequences  $\overrightarrow{S}(W) = \langle W_0^-,W_0^+, \ldots,  W_2^-, W_2^+ \rangle$ is given in Figure \ref{fig:MorseSequenceTorusSkeleton}. 
In each of the figures (a) and (b), we have: 
\begin{enumerate}
\item[-] $W_0^- = \emptyset$ and  $W_0^+ = \{ a \}$, where $a$ is the critical vertex.  
\item[-] $W_1^-$ is the tree composed of $W_0^+$, all the green vertices, and all the green edges, 
\item[-] $W_1^+$ is composed of $W_1^-$ and all the critical edges, that is, all the blue edges, 
\item[-] $W_2^-$ is composed of $W_1^+$, all the dark edges, and all the grey triangles. 
\item[-] $W_2^+ = K$, that is, $W_2^+ =  W_2^- \cup \{\ka\}$, where $\ka$ is the critical yellow triangle. 
\end{enumerate}
The two next propositions are direct consequences of the above definitions.\\

\begin{proposition} \label{pro:reg2} Let $\ms$ be an arranged Morse sequence 
and let $\overrightarrow{S}(W)$ be the skeleton sequence of $\ms$. Then each complex in  $\overrightarrow{S}(W)$ is a 
complex in  $\ms$.\\
 \end{proposition}

Note that we can also affirm that $\overrightarrow{S}(W)$ is \emph{a subsequence} of an arranged sequence~$\ms$. 
That is, $\overrightarrow{S}(W)$ can be obtained from $\ms$
by removing some elements without changing the order of the remaining elements. \\

\begin{proposition} \label{pro:reg3a}
Let $\ms(K) = \langle K_0,\ldots,K_k \rangle$ be an arranged Morse sequence.
Let $K_i = W_{p}^+$ and  $K_j = W_{p+1}^-$. Then  the sequence $\langle K_j,\ldots, K_i \rangle$ is a collapse sequence.\\
 \end{proposition}

The following result is a direct consequence of Theorem \ref{pro:reg1} and Proposition~\ref{pro:reg3a}. It  may be seen as another aspect of the fundamental collapse theorem of discrete Morse theory; see Remark 2
given in Section \ref{sec:seq}.
A specificity of this formulation is that all  regular pairs (and not only a part of them) of a given dimension
are removed by the collapse sequence.\\

\begin{theorem} \label{pro:reg3}
If $\ms$ is a Morse sequence, then $W_{p+1}^-$ collapses onto $W_{p}^+$.\\
 \end{theorem}
 
 \ELIMINE{
 \begin{proof} 
By Theorem \ref{pro:reg1}, there exists an arranged Morse sequence $\overrightarrow{V}$ such that $\ms$ and $\overrightarrow{V}$
are  equivalent. By Proposition \ref{pro:reg2}, the complexes $V_{p+1}^-$ and $V_{p}^+$ are in the sequence $\overrightarrow{V}$.
By the definition of skeletons, we have: \\
\hspace*{\fill}
$V_{p+1}^- \setminus V_{p}^+ = \{ \nu \in  \ov{W} \; | \; dim(\nu) = p+1 \} 
\cup \{ \nu \in \un{W} \; | \; dim(\nu) = p \}$. 
\hspace*{\fill} \\
It means that the sequence in $\overrightarrow{V}$ between $V_{p}^+$ and $V_{p+1}^-$ is made solely of elementary expansions.
Thus $V_{p+1}^-$ collapses onto $V_{p}^+$.
Since $\ms$ and $\overrightarrow{V}$
are  equivalent, we have $\overrightarrow{S}(V) = \overrightarrow{S}(W)$. Therefore $W_{p+1}^-$ collapses onto $W_{p}^+$.
\qed
\end{proof}
}

We now  give some properties  of reference and coreference maps. 
Theorems \ref{pro:reg5} and \ref{pro:injec}
will be essential for some crucial proofs in the sequel of this paper.
These theorems are obtained thanks to arranged sequences that may be derived from arbitrary Morse sequences (Theorem \ref{pro:reg1}), and thanks to 
the skeletons which appear in such sequences (Proposition \ref{pro:reg2}). 
 
 Let  $\nu \in K^{(p)}$ such that $\nu \in W_p^+$. Thus $\nu \in \ha{W} \cup \ov{W}$. We have $\rf (\nu) = \nu$ if $\nu \in \ha{W}$,
 and $\rf (\nu) = 0$ if $\nu \in \ov{W}$. This leads us to the following:\\
 
\begin{proposition} \label{pro:reg4} Let $\ms(K)$ be a Morse sequence and let $c \in K [p]$. 
\begin{enumerate}
\item If $c \subseteq W_p^+$, then $\rf (c) = c \cap \ha{W}$. 
\item If $c \subseteq K \setminus W_p^-$, then $\corf (c) = c \cap \ha{W}$. \\
\end{enumerate}
 \end{proposition}

\begin{theorem} \label{pro:reg5} Let $\ms(K)$ be a Morse sequence. 
\begin{enumerate}
 \item Let $z,z' \in Z_p(K)$. If $\rf (z) = \rf (z')$ and $z,z' \subseteq W_p^+$, then $z = z'$. 
 \item Let $z,z' \in Z^p(K)$. If $\corf (z) = \corf (z')$ and $z,z' \subseteq K \setminus W_p^-$, then $z = z'$.
 \end{enumerate}
 \end{theorem}
 
\begin{proof} 
By  Theorem \ref{pro:reg1} and by Proposition \ref{pro:equi}, we may suppose that
$\ms(K)$ is an arranged sequence, we set $\ms(K) = \langle K_0,\ldots,K_k \rangle$. \\
1) We will first prove that $z = 0$ whenever $z \in Z_p(K)$ and $z \subseteq W_p^-$. 
Since $W_0^- = \emptyset$, the property is trivially true whenever $p=0$.
Suppose $p \geq 1$. 
By Proposition \ref{pro:reg2}, there exist $K_i$ and $K_j$ such that
$K_i = W_{p-1}^+$ and $K_j = W_{p}^-$. Since $dim(K_i) \leq p-1$, the property is trivially true whenever 
$z \subseteq K_i$. Suppose the property is true whenever 
$z \subseteq K_{l-1}$, with $i \leq l-1 \leq j-1$, and let $z \in Z_p(K)$ with $z \subseteq K_l$. 
By Proposition~\ref{pro:reg3a} we have $K_l = K_{l-1} \cup \{\sig,\tau\}$, where $(\sig,\tau)$ is a free pair
and $dim(\tau) = p$.
But it is not possible that $z$ contains the simplex $\tau$ otherwise we would have $\sig \in \partial(z)$ and $z$ would not be a cycle.
Thus, $z \subseteq K_{l-1}$. By the induction hypothesis, it means that $z =0$. 
Therefore the property is true whenever $z \subseteq K_j$. Since $K_j = W_{p}^-$, we obtain the desired conclusion. \\
2) Now, suppose $z \in Z_p(K)$, $z \subseteq W_p^+$ and $\rf (z) = 0$. By  Proposition \ref{pro:reg4}, 
it means that $z \subseteq W_p^-$. By the result obtained in 1), we deduce that $z=0$. 
Let $z,z' \in Z_p(K)$ such that
$z,z' \subseteq W_p^+$, and $\rf (z) = \rf (z')$. We have $\rf (z+ z')= 0$. By the previous result, this gives $z = z'$.
\end{proof}

By the previous statement, two distinct $p$-cycles in $W_p^+$ must have distinct references.
Of course this property is not true if we consider arbitrary $p$-cycles.
For example, let us consider the $1$-cycle $z$ that is the boundary of the dunce hat
depicted Figure \ref{fig:MorseRefDunceHat} (b). We have $\rf(z) = b$.
But the $1$-cycle $z' = \pl(c)$, the boundary of the critical simplex $c$, 
is made of three simplices $\nu$ such that $\rf(\nu) = b$. Thus we have 
$z \not = z'$ but $\rf(z) = \rf(z') = b$. 

In the next result, we consider the general case where the $p$-cycles are not necessarily in $W_p^+$.
We write $z \sim z'$ if $z$ and $z'$ are two cycles (or two cocycles) that are in the same homology (or cohomology) class.\\

 \begin{theorem} \label{pro:injec}
Let $\ms(K)$ be a Morse sequence. 
\begin{enumerate}
 \item Let $z,z' \in Z_p(K)$. If $\rf (z) = \rf (z')$, then $z \sim z'$. 
\item Let $z,z' \in Z^p(K)$. If $\corf (z) = \corf (z')$, then $z \sim z'$.
\end{enumerate}
 \end{theorem}

\begin{proof}
Again, by  Theorem \ref{pro:reg1} and by Proposition \ref{pro:equi}, we may suppose that
$\ms(K)$ is an arranged sequence, we set $\ms(K) = \langle K_0,\ldots,K_k \rangle$. 
We will prove that we have $z \sim 0$ whenever $z \in Z_p(K)$ and $\rf (z) =0$. The result will follow by linearity. \\
Let $z \in Z_p(K)$. We first observe that $z \subseteq W_{p+1}^-$. 
By Proposition \ref{pro:reg2}, there exist $K_i$ and $K_j$ such that $K_i = W_{p}^+$ and $K_j = W_{p+1}^-$. 
By Theorem \ref{pro:reg5}, the property is true whenever $z \subseteq K_i$. 
Suppose it is true for any $z \subseteq K_{l-1}$, with $i \leq l-1 \leq j-1$, and let $z \in Z_p(K)$ such that $z \subseteq K_l$ and $\rf (z) =0$. 
By Proposition \ref{pro:reg3a} we have $K_l = K_{l-1} \cup \{\sig,\tau\}$, where $(\sig,\tau)$ is a free pair. 
We have $dim(\tau) = p+1$, thus the simplex $\tau$ is not in $z$. 
Also by the induction hypothesis, we have $z \sim 0$ whenever $\sig \not\in z$. \\
Suppose now that $\sig \in z$. Let $z' = z + \partial (\tau)$. We have $\rf(z') = \rf(z) + \rf(\partial (\tau)) = \rf(z) = 0$.
We have $z' \subseteq K_{i-1}$. Therefore, by the induction hypothesis, we have $z' \sim 0$. But $\partial(\tau)$ is a boundary, thus $z \sim z'$. 
Therefore $z \sim 0$. 
\end{proof}

It should be noted that the converse of the previous statement is, in general, not true. For example, let us consider again the cycle $z$ that is the boundary of the dunce hat
depicted in Figure \ref{fig:MorseRefDunceHat} (b). We have $\rf(z) = b$, but $z \sim 0$ and $\rf(0) = 0$.


\section{The extension and coextension maps}
\label{sec:ext}

In this section, we introduce two maps $\tirf$ and $\ticorf $ which, at the level of homology, may be seen as the inverses of the maps
$\rf$ and $\corf$. Each couple $(\rf,\tirf)$ and $(\corf,\ticorf)$ allows us to establish a relationship between 
a complex and its critical complex. In particular, this relation includes the equality of their homology.  \\

\begin{definition} \label{def:ext}
Let $\ms(K)$ be a Morse sequence.  
Let $\tirf_p$ and $\ticorf_p$ be the two linear maps 
$\tirf_p:$ $\ha{W}[p] \rightarrow  K[p]$ and $\ticorf_p:$ $\ha{W}[p] \rightarrow  K[p]$
such that, for each $\kappa \in \ha{W}$, \\
\hspace*{\fill} 
$\tirf (\kappa) = \{ \nu \in K \; | \; \kappa \in \corf(\nu) \}$
and $\ticorf (\kappa) = \{ \nu \in K \; | \; \kappa \in \rf(\nu) \}$.
\hspace*{\fill} \\
$\widetilde{\rf}$ and $\widetilde{\corf}$ are, respectively, the \emph{extension and
the coextension map of~$\ms$}. If $c \in \ha{W} [p]$, 
$\widetilde{\rf}_p (c)$ and $\widetilde{\corf}_p (c)$ are, respectively, the \emph{extension and
the coextension of $c$ (for~$\ms$)}.\\
\end{definition}


\noindent
In other words, if $\kappa \in \ha{W}^{(p)}$, we have: \\
\hspace*{\fill}
$\nu \in \tirf (\kappa)$ if and only if $\kappa \in \corf(\nu)$ and $\nu \in \ticorf (\kappa)$ if and only if $\kappa \in \rf(\nu)$.
\hspace*{\fill} \\
Also, the two maps $\tirf$ and $\ticorf$ may be described by $\corf$- and $\rf$-paths: 
\begin{enumerate}
\item[-] We have $\nu \in  \tirf (\kappa)$ if and only if there exists a $\corf$-path from $\ka$ to $\nu$.
\item[-] We have $\nu \in  \ticorf (\kappa)$ if and only if there exists a $\rf$-path from $\nu$ to $\kappa$. 
\end{enumerate}
An illustration of two extension maps is directly given by the coreference maps of the torus and the dunce hat 
depicted in Figure \ref {fig:MorseRefTorus} (c) and \ref{fig:MorseRefDunceHat} (c). Another example is given in
Figure  \ref{fig:Ext} where the effect of the modulo 2 arithmetic is emphasized.  

\begin{figure*}[tb]
    \centering
    \begin{subfigure}[t]{0.25\textwidth}
        \centering
        \includegraphics[width=0.7\textwidth]{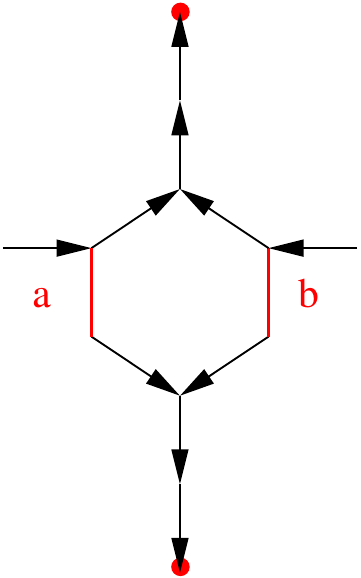}
        \caption{}
    \end{subfigure}%
    ~
    \begin{subfigure}[t]{0.25\textwidth}
        \centering
        \includegraphics[width=0.7\textwidth]{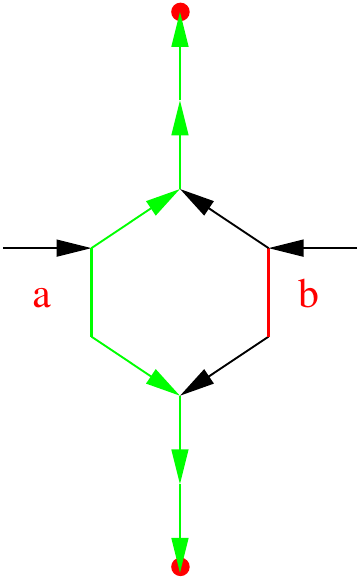}
        \caption{}
    \end{subfigure}%
     ~
    \begin{subfigure}[t]{0.25\textwidth}
        \centering
        \includegraphics[width=0.7\textwidth]{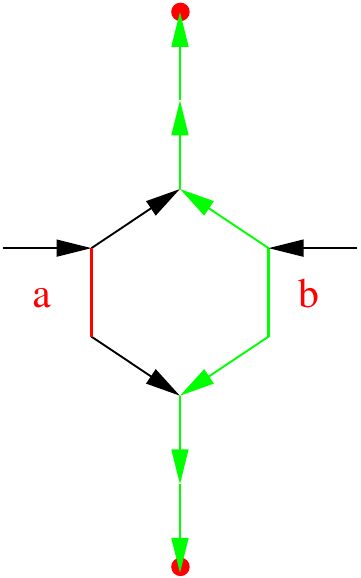}
        \caption{}
    \end{subfigure}%
     ~
    \begin{subfigure}[t]{0.25\textwidth}
        \centering
        \includegraphics[width=0.7\textwidth]{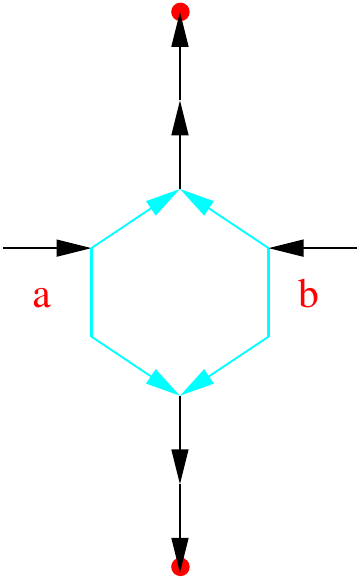}
        \caption{}
    \end{subfigure}%
 \caption{(a) The gradient vector field of a Morse sequence with 4 critical faces. 
  (b)  The extension of the face $a$. (c) The extension of the face $b$. (d) The extension of $a+b$.} 
 \label{fig:Ext}
\end{figure*}

\ELIMINE{
\begin{definition} \label{def:ext}
Let $\ms(K)$ be a Morse sequence. The two linear maps: \\ 
\hspace*{\fill}
$\tirf_p:$ $\msc[p] \rightarrow  K[p]$ and $\ticorf_p:$ $\msc[p] \rightarrow  K[p]$
\hspace*{\fill} \\
are defined by  
$\tirf (\kappa) = \{ \nu \in K \; | \; \corf(\nu) = \kappa \}$
and $\ticorf (\kappa) = \{ \nu \in K \; | \; \rf(\nu) = \kappa \}$,
 where $\kappa \in \msc(p)$. 
If $c \in  \msc[p]$, then $\widetilde{\rf}_p (c)$ is the \emph{$\rf$-extension of $c$} and
$\widetilde{\corf}_p (c)$ is the \emph{$\corf$-extension of $c$}.
\end{definition}
}

With the notion of an extension map, we retrieve the closure operator $\varphi$ introduced in \cite{koz20}, Chapter 12.
This fact can be checked by using Proposition 12.7 of  \cite{koz20} and Theorem \ref{pro:grad1} of this paper.
It has been proved that the image of the operator $\varphi$ induces a chain complex 
which leads to 
an isomorphism with the critical complex, see Chapter 13 of \cite{koz20}. Perhaps
the main difference between $\tirf$ and $\varphi$ is that the map 
$\tirf$ is defined with the coreference map
 $\corf$, while $\varphi$ is built with a decomposition of a graph, see Definition 12.3 of \cite{koz20}. 
 In the following, we will see that the link between $\tirf$ and
 $\corf$ allows us to highlight the strong relationships
 between the maps $\tirf$, $\rf$, $\plc$ and, in a dual way, between the maps $\ticorf$, $\corf$, $\dlc$. We will recover, 
 in a different manner, and without the image of $\tirf$, an isomorphism with the critical complex. In the next section, 
 we will use the chain complex induced by the image of $\tirf$ for an equivalence with the flow complex.

Let $\kappa \in \ha{W}$. We have $\corf(\kappa) = \rf(\kappa) = \kappa$,
therefore $\kappa \in \tirf(\kappa)$ and $\kappa \in \ticorf(\kappa)$.
Also, if $\kappa' \in \ha{W}$ and $\kappa \in \corf(\kappa')$ or $\kappa \in \rf(\kappa')$,
we have $\kappa = \kappa'$ (since $\corf(\kappa') = \rf(\kappa') = \kappa'$).
Furthermore: 
\begin{enumerate}
\item[-] if $\nu \in \tirf(\kappa)$, then $\nu$ cannot be lower regular (since $\corf(\nu) \not= 0)$. 
\item[-] if $\nu \in \ticorf(\kappa)$, then $\nu$ cannot be upper regular (since $\rf(\nu) \not= 0)$. 
\end{enumerate}
Thus, we have the two following results. \\

 \begin{proposition} \label{pro:ext1}
 Let $\ms$ be a Morse sequence on $K$ and let $\kappa \in \ha{W}$. Then: 
 \begin{enumerate}
\item $\tirf (\kappa)$ and $\ticorf (\kappa)$ each contains a unique critical simplex which is precisely $\kappa$.
\item If $\nu \in \tirf(\kappa)$ and $\nu \not= \kappa$, then $\nu$ is 
upper regular for $\ms$.  That is $\tirf (\kappa)+ \kappa \subseteq \overline{W}$.
\item If $\nu \in \ticorf(\kappa)$ and $\nu \not= \kappa$, then $\nu$ is 
lower regular for $\ms$.  That is $\ticorf (\kappa)+ \kappa \subseteq \underline{W}$. \\
\end{enumerate}
 \end{proposition}
 
  \begin{proposition} \label{pro:ext10}
  Let $\ms$ be a Morse sequence and let $c \in \ha{W}[p]$. We have: \\
  \hspace*{\fill}
$\tirf_{p} (c) \subseteq W_{p}^+$ and $\ticorf_{p} (c) \subseteq K \setminus W_{p}^-$.
 \hspace*{\fill} \\
   \end{proposition}
 
 Let $\kappa \in \ha{W}$. We have $\rf(\nu) =0$ whenever $\nu$ is upper regular. Thus, by Proposition~\ref{pro:ext1}, we have 
 $\rf (\widetilde{\rf} (\kappa)) = \kappa$. If we denote by $Id_S$ the identity map on a set $S$, we obtain by linearity: \\

  \begin{proposition} \label{pro:ext8}
  Let $\ms$ be a Morse sequence. We have: \\
\hspace*{\fill}
$\rf_{p} \circ \widetilde{\rf}_p = Id_{\ha{W} [p]}$ and
$\corf_{p} \circ \widetilde{\corf}_p  = Id_{\ha{W} [p]}$.
\hspace*{\fill}\\
 \end{proposition}

 Let $\kappa \in  \ha{W}$. With Propositions \ref{pro:cons1} and \ref{pro:ext1}, we derive
 $\rf(\partial(\tirf (\kappa)+ \kappa)) = 0$.  
It means that  $\rf(\partial(\tirf (\kappa))) = \rf(\partial(\kappa))$ where the last expression corresponds to the definition of the map $\plc$. 
By linearity we obtain:  \\
 
 \begin{proposition} \label{pro:ext2}
  Let $\ms$ be a Morse sequence. We have: \\
\hspace*{\fill}
$\rf_{p-1} \circ \partial_p \circ \widetilde{\rf}_p = \plc_p  $ and
$\corf_{p+1} \circ \delta_p \circ \widetilde{\corf}_p  = \dlc_p$
\hspace*{\fill} \\
 \end{proposition}
 
 \ELIMINE{
  \begin{lemma} \label{lem:ext3}
 Let $\ms$ be a Morse sequence on $K$. Let $\kappa \in \ha{W}$ and let $(\sig, \tau)$ be a free pair for $\ms$. We set: \\
 \hspace*{\fill} 
 $S = \{\nu \in \pl(\tau) \; | \; \kappa \in \rf(\nu) \}$ and 
 $T = \{\nu \in \delta(\sig) \; | \; \kappa \in \corf(\nu) \}$. 
 \hspace*{\fill} \\
Then, each number $card(S)$ and $card(T)$ is even.
 \end{lemma}
 
\begin{proof}

Let us consider the set $S' = \{\nu \in \pl(\tau)+ \sig \; | \; \kappa \in \rf(\nu) \}$.
The simplex $\sig$ is not in $S'$. 
By definition of $\rf$, we have
$ \rf(\sig) = \rf(\partial(\tau) + \sig)$. Therefore we have $\kappa \in \rf(\sig)$
if and only if $card(S')$ is odd. Thus $\sig \in S$ if and only if $card(S')$ is odd.
If $\sig \in S$, we have $S = S' \cup \{\sig\}$ and $card(S)$ is even since $\sig \not\in S'$.
If $\sig \not\in S$, we have $S= S'$ and $card(S)$ is also even.
\qed
\end{proof}
}

Let us consider the critical face $c$ of the dunce hat depicted in Figure \ref{fig:MorseRefDunceHat} (c). 
The extension of $c$, that is the set $\widetilde{\rf}(c)$, is composed of all yellow faces. 
By Prop. \ref{pro:ext1}, if $\nu \in \widetilde{\rf}(c)$, 
 then either $\nu = c$, or $\nu$ is upper regular for $\ms$.
Now, let us consider the boundary of $\tirf(c)$. The set $\partial (\tirf(c))$ is composed of all the blue segments.
By examining the skeleton $W_1^+$ of Fig. \ref{fig:MorseSequenceTorusSkeleton} (b), it can be checked that this set is also
composed solely of critical and upper regular faces. This illustrates the following result. \\


 \begin{proposition} \label{pro:ext4}
 Let $\ms$ be a Morse sequence and let $\kappa \in \ha{W}$. 
 \begin{enumerate}
\item If $\nu \in \partial (\tirf(\kappa))$, then $\nu$ is either critical or
upper regular for $\ms$.
\item If $\nu \in \delta (\ticorf(\kappa))$, then $\nu$ is either critical or
lower regular for $\ms$.
\end{enumerate}
 \end{proposition}

\begin{proof} Let $\kappa \in \ha{W}$ and let $\sig$ be a
simplex that is lower regular for $\ms$.
Let 
$T = \{\nu \in \tirf(\kappa)  \; | \; \sig \in \pl(\nu) \}$. 
By the definition of the boundary operator,
we have $\sig \in \partial (\tirf(\kappa))$ if and only if $card(T)$ is odd. 
But we may also write $T = \{\nu \in \delta(\sig) \; | \; \kappa \in \corf(\nu) \}$. 
We observe that we have $\kappa \in \corf(\delta(\sig))$ if and only if $card(T)$ is odd.
But, since $\sig$ is lower regular, we must have $\corf(\delta(\sig)) = 0$ (Proposition \ref{pro:cons1}). 
Therefore $card(T)$ must be even.
It means that $\partial (\tirf(\kappa))$ cannot contain a simplex that is lower regular for $\ms$.
\end{proof}

\ELIMINE{
Another way of stating Proposition \ref{pro:ext4} is to say that, if $\kappa \in \ha{W}^{(p)}$, then: \\
1) $\partial (\tirf(\kappa)) \subseteq W_{p-1}^+$. \\
!!! 2) $\delta (\ticorf(\kappa)) \subseteq W^{p+1}$.
}

\begin{proposition} \label{pro:ext5}
 Let $\ms$ be a Morse sequence and let $\kappa \in \ha{W}$. We have: \\
 \hspace*{\fill}  $\plc(\kappa) = \partial (\tirf(\kappa)) \cap \ha{W}$ and $\dlc(\kappa) = \delta (\ticorf(\kappa)) \cap \ha{W}$.  \hspace*{\fill} 
 \end{proposition}
 
 \begin{proof} By Proposition \ref{pro:ext2}, we have
 $\plc(\kappa) = \rf(\partial(\tirf (\kappa)))$. Let $\nu \in K$. We can write: \\
 \hspace*{\fill}
 $\nu \in \rf(\partial(\tirf (\kappa))) \Leftrightarrow card \{\nu' \; | \; \nu' \in \partial(\tirf (\kappa))$ and $\nu \in \rf(\nu') \}$ is odd. 
 \hspace*{\fill} \\
 If $\nu'$ is an upper regular simplex of $\partial(\tirf (\kappa))$, we have $\rf(\nu') = 0$. Furthermore, by 
Proposition \ref{pro:ext4}, $\partial(\tirf (\kappa))$ does not contain any lower regular simplex. Therefore,
if $\nu' \in \partial(\tirf (\kappa))$ and $\nu \in \rf(\nu')$, we must have $\nu' \in \ha{W}$. In this case, we also have 
$\rf(\nu') = \nu'$. We obtain $\nu = \nu'$. Thus, we can write: \\
\hspace*{\fill}
$\nu \in \rf(\partial(\tirf (\kappa))) \Leftrightarrow card \{\nu \; | \; \nu \in \partial(\tirf (\kappa)) \cap \ha{W} \}$ is odd, 
\hspace*{\fill} \\
This last number must be equal to 1.
That is, we have $\nu \in \rf(\partial(\tirf (\kappa)))$ if and only if $\nu \in \partial(\tirf (\kappa)) \cap \ha{W}$, which gives the desired result.
\ELIMINE{
: \\
$\nu \in \rf(\partial(\tirf (\kappa)))$ if and only $card \{\nu' \; | \; \nu' \in \partial(\tirf (\kappa)) \cap \msc$ and $\nu \in \rf(\nu') \}$ is odd. \\

Therefore, we have $\nu \in \rf(\partial(\tirf (\kappa)))$ if and only if $\nu$ is critical for $\ms$. \\
But 

$\nu \in \rf(\partial(\tirf (\kappa)))$ if and only $card \{\nu' \; | \; \nu' \in \partial(\tirf (\kappa)) \cap \msc$ and $\nu = \nu' \}$ is odd. \\
$\nu \in \rf(\partial(\tirf (\kappa)))$ if and only $card \{\nu \; | \; \nu \in \partial(\tirf (\kappa)) \cap \msc \}$ is odd. \\

Therefore, the set $\rf(\partial(\tirf (\kappa)))$ contains the critical simplices that are precisely in $\partial(\tirf (\kappa))$.
That is, we have $\rf(\partial(\tirf (\kappa))) = \partial (\tirf(\kappa)) \cap \msc$.
} 
\end{proof}


\begin{proposition} \label{pro:ext6}
Let $\ms$ be a Morse sequence on $K$. 
\begin{enumerate}
\item If $z \in Z_p(\ha{W})$, then $\tirf_p (z) \in Z_p(K)$. 
\item If $z \in Z^p(\ha{W})$, then $\ticorf_p (z) \in Z^p(K)$. 
\end{enumerate}
\end{proposition}

 \begin{proof}
 If $z \in Z_0(\ha{W})$, then $z \in K[0]$, $\tirf_0 (z) = z$, and $\partial(z) = 0$.  
 Thus the property is true for $p=0$.
Suppose the property is true for any $z \in Z_{p-1}(\ha{W})$, with $p \geq 1$. \\
1) Let $\kappa \in \ha{W}$ with $\kappa \in K^{(p)}$. 
by Proposition \ref{pro:ext8}, we have $\plc(\kappa) = \rf_{p-1}(\tirf_{p-1} (\plc(\kappa)))$. 
With Prop. \ref{pro:ext2} we obtain : \\
\hspace*{\fill}
 $\plc(\kappa) = \rf_{p-1}(\partial_p (\tirf(\kappa)))$ and $\plc(\kappa) = \rf_{p-1}(\tirf_{p-1} (\plc(\kappa)))$. 
 \hspace*{\fill} \\
Now, let us consider the two expressions $\partial_p (\tirf(\kappa))$ and $\tirf_{p-1} (\plc(\kappa))$. \\
- We have $\partial_p (\tirf(\kappa)) \in Z_{p-1}(K)$.
Since $\plc(\kappa) \in Z_{p-1}(\ha{W})$, by the induction hypothesis, we also have
$\tirf_{p-1} (\plc(\kappa)) \in Z_{p-1}(K)$. \\
- By Proposition \ref{pro:ext4}, we have $\partial_p (\tirf(\kappa)) \subseteq  W_{p-1}^+$. Furthermore, 
by Proposition \ref{pro:ext1}, we have $\tirf_{p-1} (\plc(\kappa)) \subseteq W_{p-1}^+$. \\
Thus, by considering the two above expressions of $\plc(\kappa)$,
Theorem \ref{pro:reg5} leads to: \\
\hspace*{\fill} $\partial_p (\tirf(\kappa)) = \tirf_{p-1} (\plc(\kappa))$. \hspace*{\fill}\\
2) Let $z \in Z_p(\ha{W})$, that is $\plc_p  (z) = 0$. We have
$\partial_p (\tirf_p (z)) = \sum \{ \partial_p (\tirf (\kappa)) \; | \; \kappa \in z \}$. \\
With the above result we obtain:  \\
$\partial_p (\tirf_p (z)) = \sum_{\kappa \in z} \tirf_{p-1} (\plc(\kappa))$
$= \tirf_{p-1} (\sum_{\kappa \in z}  \plc(\kappa)) = \tirf_{p-1} (\plc_p  (z)) = 0$. 
\end{proof}

  \begin{proposition} \label{pro:ext7}
  Let $\ms$ be a Morse sequence. 
  \begin{enumerate}
\item If $z \in Z_p(K)$, then $\tirf_p(\rf_{p} (z)) \sim z$. 
\item If $z \in Z^p(K)$, then $\ticorf_p(\corf_{p} (z)) \sim z$. 
\end{enumerate}
\end{proposition}

 \begin{proof}
Let $z \in Z_p(K)$.
By Theorem \ref{pro:label3}, $\rf_{p} (z) \in Z_p(\ha{W})$, and by Proposition \ref{pro:ext6} $\tirf_p(\rf_{p} (z)) \in Z_p(K)$.
By Proposition \ref{pro:ext8}, we have 
$\rf_{p} (\tirf_p (\rf_{p} (z))) = \rf_{p} (z)$. Now, by applying Theorem \ref{pro:injec}, we obtain $\tirf_p(\rf_{p} (z)) \sim z$. 
\end{proof}

 \begin{theorem} \label{pro:chainmap2}
 Let $\ms$ be a Morse sequence. We have: \\
\hspace*{\fill}
$\partial_p \circ \tirf_p = \tirf_{p-1} \circ \plc_p  $ and
$\dl_p \circ \ticorf_p = \ticorf_{p+1} \circ \dlc_p$ 
\hspace*{\fill}
 \end{theorem}
 
\begin{proof} Let $\kappa \in \ha{W}$ with $\kappa \in K^{(p)}$. \\
 - By Proposition \ref{pro:ext4}, we have $\partial_p (\tirf(\kappa)) \subseteq W_{p-1}^+$. By Proposition \ref{pro:reg4},
we obtain $\rf_{p-1}(\partial_p (\tirf(\kappa))) = \partial_p (\tirf(\kappa)) \cap \ha{W}$. 
Thus, by Proposition \ref{pro:ext5},\\
 \hspace*{\fill}
  $\rf_{p-1}(\partial_p (\tirf(\kappa))) = \plc(\kappa)$. 
   \hspace*{\fill} \\
 - By Proposition \ref{pro:ext8}, we have
$\rf_{p-1}(\tirf_{p-1} (\plc (\kappa)) ) = \plc(\kappa)$. Thus: \\
 \hspace*{\fill}
 $\rf_{p-1}(\partial_p (\tirf(\kappa))) = \rf_{p-1}(\tirf_{p-1} (\plc (\kappa)) )$.
  \hspace*{\fill} \\
Furthermore,  by Proposition \ref{pro:ext6}, we have $\tirf_{p-1} (\plc (\kappa) ) \in Z_{p-1} (K)$ and,
by Proposition \ref{pro:ext10}, $\tirf_{p-1} (\plc (\kappa) ) \subseteq W_{p-1}^+$. 
Thus, by Theorem  \ref{pro:reg5}, we obtain: \\
 \hspace*{\fill} 
$\partial_p (\tirf(\kappa)) = \tirf_{p-1} (\plc (\kappa)) $,
 \hspace*{\fill} \\
which gives the result by linearity. 
\end{proof}

Thus, the previous theorem is a counter-part of Theorem \ref{pro:label3} given for $\rf$ and $\corf$:
  the map $\tirf$ is a chain map from the chain complex $(\ha{W}[p], \plc_p)$
to  the chain complex  $(K[p], \partial_p)$, and 
the map $\ticorf$ is a cochain map from $(\ha{W}[p], \dlc_p)$
to  $(K[p], \delta_p)$. 

The counter-part of $\rf^H_p$ and $\corf^H_p$ are two linear maps 
$\tirf^H_p$ and $\ticorf^H_p$ between $H_p(\ha{W})$ and $H_p(K)$,
and between $H^p(\ha{W})$ and $H^p(K)$. \\
From Propositions \ref {pro:ext8} and \ref{pro:ext7}, we deduce that: 
  \begin{enumerate}
\item $\widetilde{\rf}^H_p \circ \rf^H_{p}  = Id_{H_p(K)}$  and $\rf^H_{p} \circ \widetilde{\rf}^H_p = Id_{H_p(\ha{W})}$.
\item  $\widetilde{\corf}^H_p \circ \corf^H_{p}  = Id_{H^p(K)}$  and $\corf^H_{p} \circ \widetilde{\corf}^H_p = Id_{H^p(\ha{W})}$. 
\end{enumerate}
Therefore, we have the following result which relates the homology of a complex with the homology of 
its critical complex. \\
 
 
  \begin{theorem} \label{pro:chainmap3}
 Let $\ms$ be a Morse sequence. 
   \begin{enumerate}
\item The vector spaces $H_p(K)$ and $H_p(\ha{W})$ are isomorphic.
\item  The vector spaces $H^p(K)$ and $H^p(\ha{W})$ are isomorphic. \\
\end{enumerate}
 \end{theorem}
We have seen that the critical complex $(\ha{W}[p], \plc_p)$ is equivalent to the classical notion of a Morse complex.
Thus, the previous theorem allows us to retrieve a basic fact of discrete Morse theory.
There are several presentations of this fact, see \cite{For98a}, \cite{koz20}. \\
A specificity of the above results is that the isomorphism between $H_p(K)$ and $H_p(\ha{W})$ is given directly with the maps $\rf$ and $\tirf$.

\ELIMINE{
In the classical presentation of this fact: \\
1) First an isomorphism is given by a map between $H_p(K)$ and $H_p(\Phi)$, where $H_p(\Phi)$
corresponds to the homology of a \emph{flow complex}, see \cite{For98a}, section~7, \\
2) Then an isomorphism between $H_p(\Phi)$ and $H_p(\ha{W})$ is given; the boundary operator is 
obtained by computing the number of gradient paths, see \cite{For98a}, section~8.  
}
 

\section{The extension complex and the gradient flow}
\label{sec:gflow}

In discrete Morse theory, gradient flows and flow complexes are basic ingredients for setting the fundamental property of a Morse complex, that is, the equality of homology between a complex and its Morse complex.
 In this section, we first give the definition of chain complexes based on the images of $\ti{\rf}$ and $\ti{\corf}$, 
 which was mentioned in the last section. 
 Then,  
 we outline a deep link between reference and extension maps  of a Morse sequence, and gradient flows (Theorem \ref{th:gradi4}). 
 We conclude by showing the equivalence
 between the classical definition of a flow complex and the definition of a chain complex based on extension maps. \\
 
 

 Let $\ms$ be a Morse sequence. We write: \\
\hspace*{\fill}
 $\ov{O} [p] = \{\ti{\rf}_p (c) \; | \; c  \in \ha{W} [p] \}$ and $\un{O} [p] = \{\ti{\corf}_p (c) \; | \; c  \in \ha{W} [p] \}$. 
 \hspace*{\fill} \\
 That is,  $\ov{O} [p]$ is the image of the map $\ti{\rf}_p$, and $\un{O} [p]$ is the image of the map $\ti{\corf}_p$. \\ 
 
 The sets $\ov{O} [p]$ and $\un{O} [p]$ 
are two vector spaces. In the example of Figure \ref{fig:Ext}, $\ov{O} [1]$
is composed of four vectors: we have $\ov{O} [1] = \{0, \ti{\rf}_p (a), \ti{\rf}_p (b), \ti{\rf}_p (a+b) \}$. \\

\ELIMINE{
As indicated in the previous section,  there is a correspondence between an extension map $\ti{\rf}$ and the closure operator $\phi$ introduced in \cite{koz20}.
In this context the image of $\phi$ is used for a decomposition of the chain complex $(K[p], \partial_p)$. \\
$\ov{O} [p]$
The main differences between $\tirf$ and $\phi$ is that the map $\tirf$ is concentrated on the critical simplices and $\tirf$ is defined with the coreference map
 $\corf$. In the following, we will see that these two characteristics of $\tirf$ allow us to highlight the strong relationships
 between the maps $\tirf$, $\rf$, $\plc$, and in a dual way between the maps $\ticorf$, $\corf$, $\dlc$. \\
 }
 
Let $d \in \ov{O} [p]$ and let $c  \in \ha{W} [p]$ such that $d =  \ti{\rf} (c)$. 
We have $\pl_p (d) = \pl_p (\ti{\rf}_p (c))$ and, by Theorem \ref{pro:chainmap2}, $\pl_p (d) = \tirf_{p-1} (\plc_p (c))$. 
Since $\plc_p (c) \in \ha{W} [p-1]$,  we obtain $\pl_p (d) \in  \ov{O} [p-1]$. 
Thus, the boundary operator $\partial_p : K[p] \rightarrow K[p-1]$ can be restricted to
 $\pl_p:$ $\ov{O}[p] \rightarrow  \ov{O}[p-1]$. This leads us to the following definition: \\

\begin{definition} \label{def:flowcomp}
Let $\ms(K)$ be a Morse sequence.  We say that the chain complex $(\ov{O}[p], \pl_p)$ is the \emph{extension complex (of $\ms$)},
and the cochain complex $(\un{O}[p], \dl_p)$ is the \emph{coextension complex (of $\ms$)}. 
 We write $H_p(\ov{O})$ and $H^p(\un{O})$ for the corresponding homology vector spaces. \\
 \end{definition}

The map $\tirf_p$ can be considered as a map   $\tirf_p:$ $\ha{W}[p] \rightarrow  \ov{O} [p]$;
  with an abuse of notation, we use the same name 
for this map.
Also, we may consider the restriction of the map $\rf_p$ to $\ov{O}  [p]$, we write
$\rf_p:$ $\ov{O} [p] \rightarrow  \ha{W}[p]$; again we use the same name 
for this map. 
We proceed in the same manner with the maps $\ticorf_p$ and $\corf_p$.

By Proposition \ref{pro:ext8}, we have
$\rf_{p} \circ \widetilde{\rf}_p = Id_{\ha{W}[p]}$.  Now, let $d \in \ov{O} [p]$, we have $d =  \ti{\rf}_p (c)$, with $c  \in \ha{W} [p]$. 
Again by Proposition \ref{pro:ext8}, we have $\rf_{p} (d) = c$. Thus, $\ti{\rf}_p (\rf_{p} (d)) = d$.
Hence, $ \widetilde{\rf}_p \circ \rf_{p}  = Id_{\ov{O} [p] }$. Since $\rf_{p}$ and $\widetilde{\rf}_p$ are linear, we obtain: \\

   \begin{proposition} \label{th:gradi5}
 Let $\ms$ be a Morse sequence.  
 \begin{enumerate}
\item The map $\tirf_p:$ $\ha{W}[p] \rightarrow  \ov{O} [p]$ is a vector space isomorphism. Its inverse is given by the map 
 $\rf_p:$ $\ov{O} [p] \rightarrow  \ha{W}[p]$. 
 \item The map $\ticorf_p:$ $\ha{W}[p] \rightarrow  \un{O} [p]$ is a vector space isomorphism. Its inverse is given by the map 
 $\corf_p:$ $\un{O} [p] \rightarrow  \ha{W}[p]$. \\
 \end{enumerate} 
 \end{proposition}

By Theorem \ref{pro:label3}, the map $\rf$ is a chain map from the chain complex $(\ov{O}[p], \pl_p)$
to  $( \ha{W}[p], \plc_p )$, and 
by Theorem \ref{pro:chainmap2} the map $\tirf$ is a chain map from the chain complex $(\ha{W}[p], \plc_p)$ to
$(\ov{O}[p], \pl_p)$. As in the previous section, we obtain from Propositions \ref {pro:ext8} and \ref{pro:ext7}
two isomorphisms $\rf^H_p$ and $\tirf^H_p$ between $H_p(\ha{W})$ and $H_p(\ov{O})$: \\

  \begin{theorem} \label{th:gradi6}
 Let $\ms$ be a Morse sequence. 
  \begin{enumerate}
\item The vector spaces $H_p(\ov{O})$  and $H_p(\ha{W})$ are isomorphic. 
 \item The vector spaces $H^p(\un{O})$ and $H^p(\ha{W})$ are isomorphic. \\
 \end{enumerate} 
 \end{theorem}

\begin{figure*}[tb]
    \centering
    \begin{subfigure}[t]{0.28\textwidth}
        \centering
        \includegraphics[height=.9\textwidth]{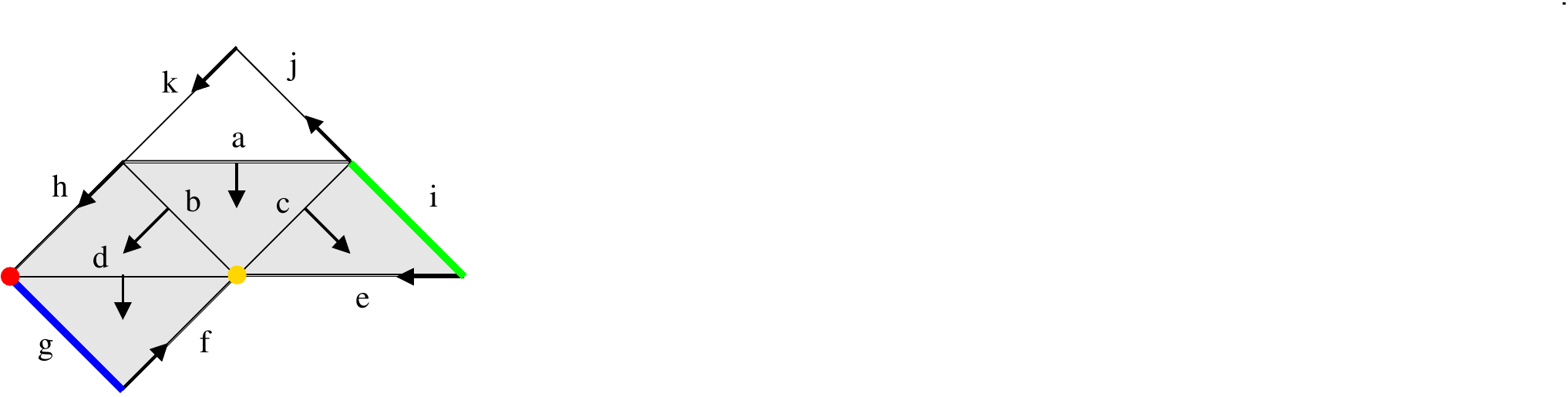}
        \caption{$\ha{W}(K)$  and $\dd{W}(K)$}
    \end{subfigure}%
    ~ \hspace*{0.4cm}
    \begin{subfigure}[t]{0.28\textwidth}
        \centering
        \includegraphics[height=.9\textwidth]{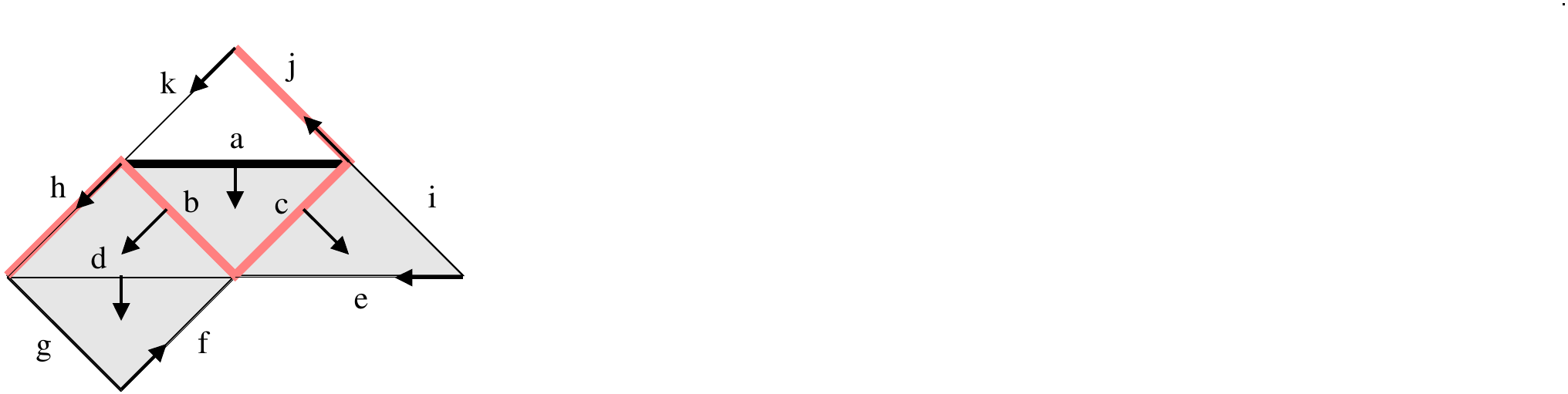}
        \caption{$\Phi(a)$}
    \end{subfigure}%
    ~ \hspace*{0.4cm}
    \begin{subfigure}[t]{0.28\textwidth}
        \centering
        \includegraphics[height=.9\textwidth]{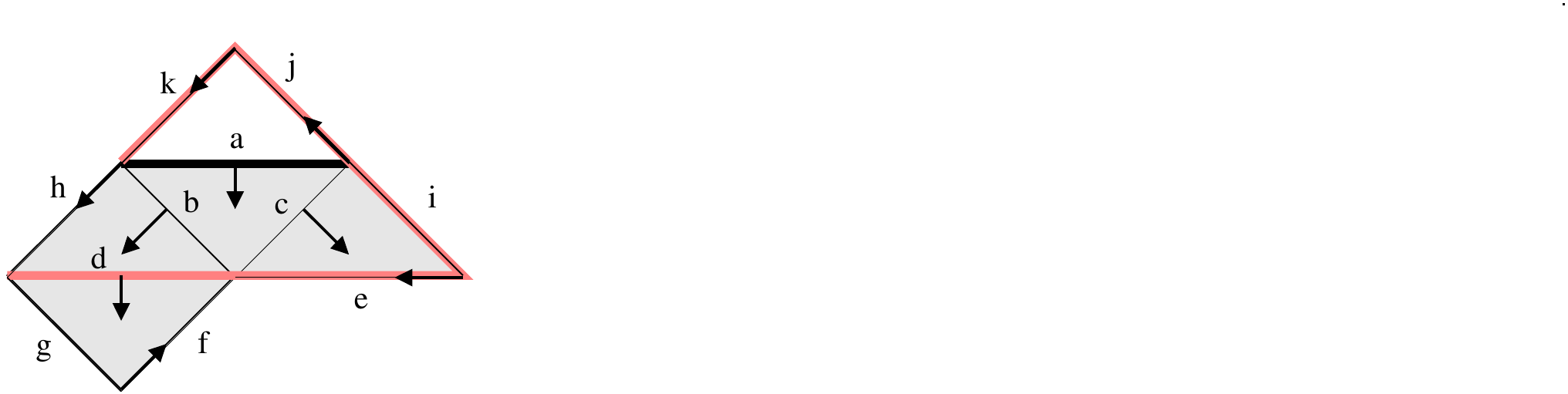}
        \caption{$\Phi^2(a)$}
    \end{subfigure}
    
    \begin{subfigure}[t]{0.28\textwidth}
        \centering
        \includegraphics[height=.9\textwidth]{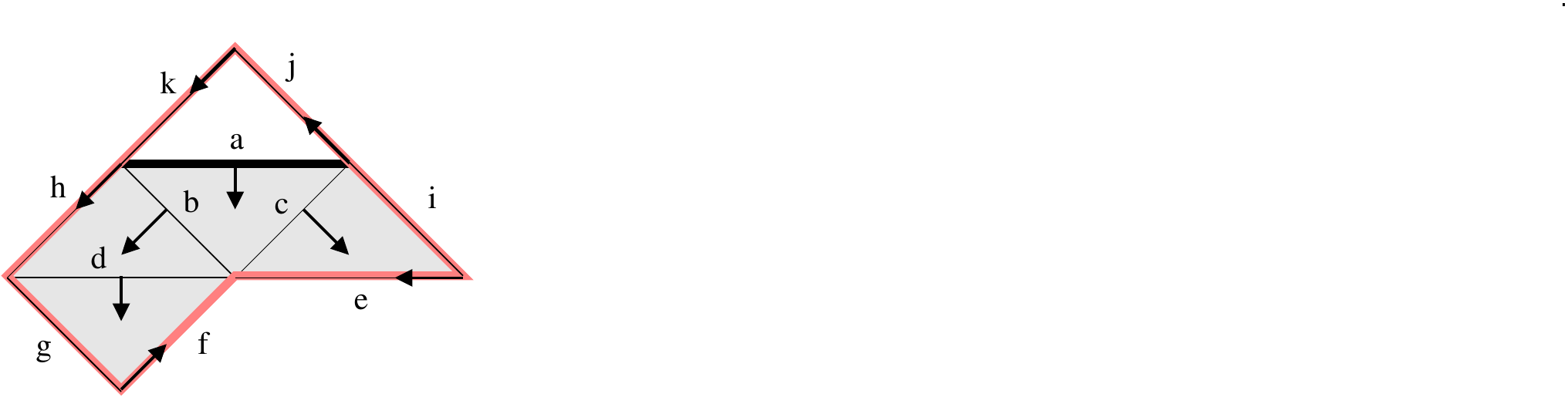}
        \caption{$\Phi^3(a)= \ov{\Phi} (a)$}
    \end{subfigure}%
    ~ \hspace*{0.4cm}
    \begin{subfigure}[t]{0.28\textwidth}
        \centering
        \includegraphics[height=.9\textwidth]{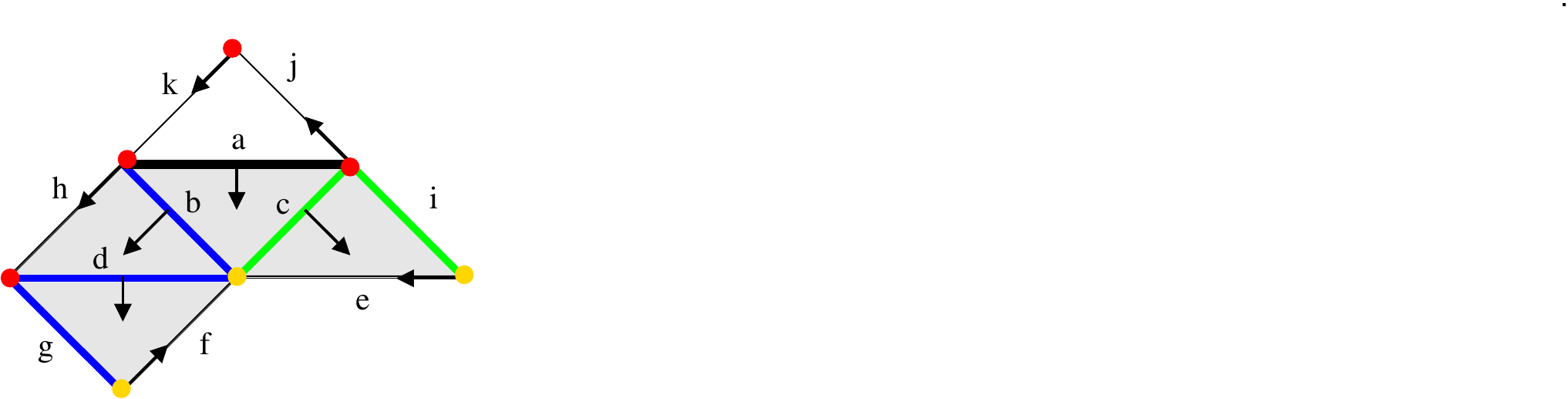}
        \caption{The map $\rf$}
    \end{subfigure}%
    ~ \hspace*{0.4cm}
    \begin{subfigure}[t]{0.28\textwidth}
        \centering
        \includegraphics[height=.9\textwidth]{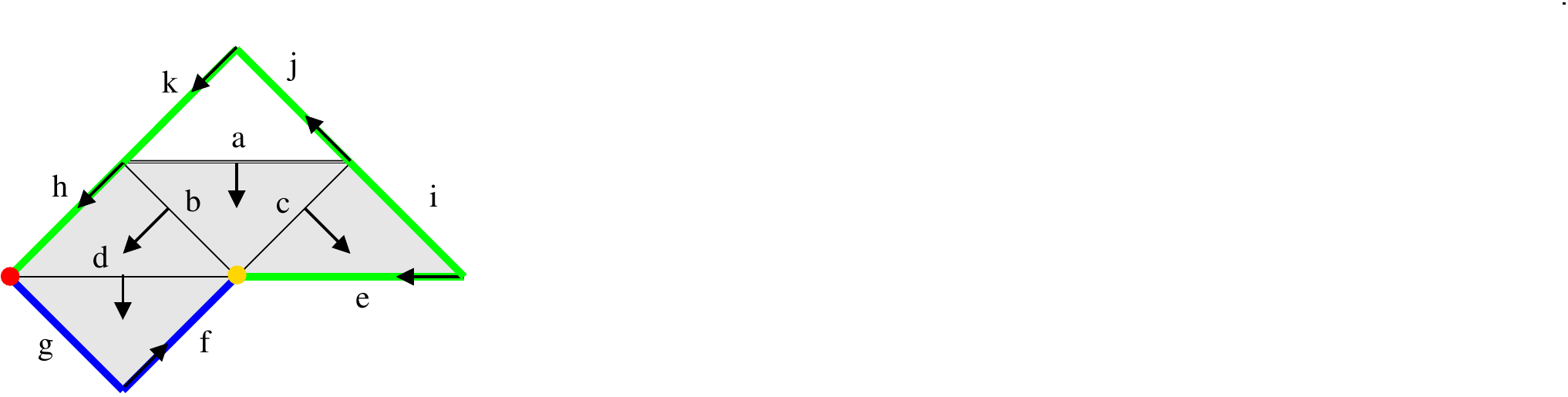}
        \caption{The map $\tirf$}
    \end{subfigure}%
 \caption{(a) The critical faces and the regular pairs of a complex $K$. (b) The first flow starting at the face $a$.
(c) The second flow.  (d) Stabilization of the flow. (e) The reference map $\rf$  of $\ms$: we have $\rf(a) = \rf(b) + \rf(c) = g+i$.
(f) The coreference map $\tirf$ of $\ms$: we have $\tirf \circ \rf (a) = \tirf(g+i) = \ov{\Phi} (a)$. 
  See text for details.
 }
 \label{fig:Flow}
\end{figure*}

Now, we recall the definition of a gradient flow, see \cite{forman2002discretecohomology}; we just adapt this definition in the context of Morse 
 sequences. See  also \cite{molina12,MOLINA12a} for an interpretation in
terms of chain homotopy equivalences. \\

Let $\ms(K)$ be a Morse sequence.  \\
Let $V_p$, $V_p^*$ be the linear maps $V_p:$ $K[p] \rightarrow  K[p+1]$, $V_p^*:$ $K[p] \rightarrow  K[p-1]$ such that: 
\begin{itemize}
\item $V(\ka) = V^*(\ka) = 0$ if $\ka$ is critical for $\ms$, 
\item $V(\sig) = \tau$, $V(\tau) = 0$, $V^*(\sig) = 0$, $V^*(\tau) = \sig$ if $(\sig, \tau)$ is a regular pair for $\ms$. 
\end{itemize}
We write $\Phi_p$, $\Phi^*_p$ for the linear maps
$\Phi_p,\Phi^*_p:$ $K[p] \rightarrow  K[p]$ 
such that: 
\begin{itemize}
\item For each $\nu \in K^{(p)}$, $\Phi_p (\nu) = \nu + \pl_{p+1}(V_p(\nu)) + V_{p-1}( \pl_p(\nu))$, 
\item For each $\nu \in K^{(p)}$, $\Phi^*_p (\nu) = \nu + \dl_{p-1}(V_p^*(\nu)) + V_{p+1}^*( \dl_p(\nu))$. 
\end{itemize}
The maps $\Phi$ and $\Phi^*$ are, respectively, the \emph{(gradient) flow and coflow}  of~$\ms$. \\
If $c \in K[p]$, then  there are integers $i,j$
such that $\Phi^{i+1} (c) = \Phi^{i} (c)$ and $(\Phi^*)^{j+1} (c) = (\Phi^*)^{j} (c)$, we write 
$\ov{\Phi}_p (c) = \Phi^{i} (c)$ and $\un{\Phi}_p (c) = (\Phi^*)^{j} (c)$.
In this way we obtain two linear maps
$\ov{\Phi}_p,\un{\Phi}_p:$ $K[p] \rightarrow  K[p]$. 

Let us consider the example of the complex $K$ shown in Fig. \ref{fig:Flow} (a), where the sets $\ha{W}(K)$ and $\dd{W}(K)$ are highlighted.
 We have $\Phi(a) = b+c+h+j$, $\Phi(b) = d$, $\Phi(c) = i+e+j$, $\Phi(h) = 0$, and $\Phi(j) = k$. Thus $\Phi^2(a) = d +i+e+j+k$,
 see  Fig. \ref{fig:Flow} (b) and (c). 
 Now, $\Phi(d) = f+g$, $\Phi(i) = i+e+j$, $\Phi(e) = 0$, and $\Phi(k) = h$. Thus $\Phi^3(a) = f+g+i+e+j+k+h$, see  Fig. \ref{fig:Flow} (d).
 Since $\Phi(f) = 0$ and $\Phi(g) = g+f$, we see that $\Phi^4(a) = \Phi^3(a)$. Therefore $\ov{\Phi} (a) = \Phi^3(a)$.

In this example, we observe a propagation effect that arises after a critical face. Specifically, we have $j \not\in \Phi(j)$, $k \not\in \Phi(k)$, and $h \not\in \Phi(h)$; 
nevertheless, since the critical face $i$ is in $\ov{\Phi} (a)$, we have $j,k,h \in \ov{\Phi} (a)$. Loosely speaking, a flow that starts at a face can be decomposed into a flow that reaches certain critical faces 
and a flow that stabilizes after these faces. This observation may help in understanding the following definition and properties. 

Let $\pi = \langle \pi_0, \ldots \pi_k \rangle$ be a sequence where each $\pi_i$ is either 
a gradient or a cogradient path in $\ms(K)$. 
We say that $\pi$ is a \emph{composite (gradient) path (from $\nu$ to $\mu$)} if: 
  \begin{enumerate}
\item The path $\pi_0$ begins at $\nu$ and the path $\pi_k$ ends at $\mu$, 
\item For each $i \in [1,k]$,  
the path  $\pi_i$ begins where $\pi_{i-1}$ ends. 
\end{enumerate}
We say that $\pi$ is \emph{nontrivial} if $\pi$ contains more than one simplex. 
  
For example, if $\pi_0 = \langle \sig_0,\tau_0,\nu \rangle$  is a gradient path and if $\pi_1 = \langle \nu, \sig_1,\tau_1 \rangle$ 
is a cogradient path, then the sequence $\pi = \langle \pi_0, \pi_1 \rangle = \langle \sig_0,\tau_0, \nu, \sig_1,\tau_1 \rangle$ is a composite path
from $\sig_0$ to $\tau_1$. Note that we have $dim(\sig_0) = dim(\tau_1)$. 

Let $\nu, \mu \in K^{(p)}$. We let $\Omega_\nu^\mu$ denote the set of all composite paths from $\nu$ to $\mu$ which include at least one 
critical simplex. Then, we have the following characterization of the maps $\ov{\Phi}$ and $\un{\Phi}$. \\

  \begin{theorem}[from  \cite{forman2002discretecohomology}] \label{th:gradi2}
 Let $\ms$ be a Morse sequence on $K$ and let $\nu \in K$. 
  \begin{enumerate}
\item We have $\mu \in \ov{\Phi} (\nu)$ if and only if $Card(\Omega_\nu^\mu)$ is odd. 
\item We have $\mu \in \un{\Phi} (\nu)$ if and only if $Card(\Omega^\nu_\mu)$ is odd. \\
 \end{enumerate}
 \end{theorem}

\noindent
From Lemma 1.3 of \cite{forman2002discretecohomology}, we also have three basic properties of composite paths
from which we easily derive Proposition \ref{pro:gradi3}: 
\begin{enumerate}
\item[-] If $\ka,\ka' \in \ha{W}$ and $\ka \not= \ka'$, then there are no composite paths from $\ka$ to $\ka'$. 
\item[-] If $\ka \in \ha{W}$, then any composite path ending at $\ka$ is a gradient path. 
\item[-] If $\ka \in \ha{W}$, then any composite path beginning at $\ka$ is a cogradient path. \\
\end{enumerate}


\begin{proposition} \label{pro:gradi3}
Let $\ms$ be a Morse sequence on $K$. Let $\nu, \mu \in K^{(p)}$ and let
$\pi \in \Omega_\nu^\mu$. Then there exists $\ka \in \ha{W}$, 
a gradient path $\pi'$ from $\nu$ to $\ka$, and a cogradient path $\pi''$ from $\ka$ to
$\nu$, such that $\pi = \langle \pi', \pi'' \rangle$. \\
\end{proposition}

Let $\nu, \mu \in K^{(p)}$.
Suppose $\nu \in \ha{W}$. Then, $\Omega_\nu^\mu$ is necessarily a cogradient path. By Theorem \ref{th:gradi2},
Theorem \ref{pro:grad1}, and  Definition~\ref{def:ext}, we obtain 
$\ov{\Phi}(\nu) = \ti{\rf}(\nu)$. In fact, we have the following more general result: \\

  \begin{theorem} \label{th:gradi4}
If $\ms$ is a Morse sequence, then we have: \\
\hspace*{\fill} 
$\ov{\Phi} = \ti{\rf} \circ \rf$ and $\un{\Phi} = \ti{\corf} \circ \corf$.
\hspace*{\fill}  
 \end{theorem}
 
 \begin{proof}
  Let $\nu, \mu \in K^{(p)}$.
 If $\ka \in \ha{W}$, we write $\Omega_\nu^\mu (\ka)$ for the set composed of all paths in $\Omega_\nu^\mu$ which contain $\ka$. 
 Since each path in $\Omega_\nu^\mu$ contains a single critical simplex, by Theorem \ref{th:gradi2},
 we have: \\
 \hspace*{\fill}
  $\mu \in \ov{\Phi} (\nu)$ if and only if $\sum \{Card(\Omega_\nu^\mu (\ka)) \; | \; \ka \in \ha{W} \}$ is odd.
  \hspace*{\fill} \\
 But $Card(\Omega_\nu^\mu (\ka)) = Card(\Omega_\nu^\ka) \times Card(\Omega_\ka^\mu)$. 
Therefore,  by parity considerations, we have 
$\mu \in \ov{\Phi} (\nu)$ if and only if: \\
\hspace*{\fill}
 $\sum \{Card(\Omega_\ka^\mu (\ka)) \; | \; \ka \in \ha{W}$ and $Card(\Omega_\nu^\ka)$ is odd$ \}$ is odd. 
 \hspace*{\fill} \\
By Proposition \ref{pro:gradi3}, $\Omega_\nu^\ka$ is the set of all gradient paths from $\nu$ to $\ka$. Thus, 
 by Theorem~\ref{pro:grad1}, we have $\rf(\nu) = \{\ka \in \ha{W} \; | \; Card(\Omega_\nu^\ka)$ is odd$ \}$. 
 Therefore: \\
 \hspace*{\fill}
  $\mu \in \ov{\Phi} (\nu)$ if and only if 
 $\sum \{Card(\Omega_\ka^\mu (\ka)) \; | \; \ka \in \rf(\nu) \}$ is odd.
 \hspace*{\fill} \\
 By Proposition \ref{pro:gradi3}, $\Omega^\mu_\ka (\ka)$ is the set of all cogradient paths from $\ka$ to $\mu$. 
 By Definition~\ref{def:ext}, and by Theorem \ref{pro:grad1},  we obtain: \\
 \hspace*{\fill}
 $\mu \in \ov{\Phi} (\nu)$ if and only if $\mu \in \ti{\rf} (\rf (\nu))$.
 \hspace*{\fill} 
 \ELIMINE{
 1) Let $\nu \in K^{(p)}$. 
 Let $\ha{C} = \{\ka \in \ha{W} \; | \; \kappa \in  \ov{\Phi} (\nu) \}$. By Proposition \ref{pro:gradi3},
if $\ka \in \ha{W}$, then $\Omega_\nu^\ka$ is the set composed of all gradient paths from $\nu$ to $\ka$. 
By Theorem \ref{th:gradi2}, we have $\ka \in \ha{C}$ if and only if $Card(\Omega_\nu^\ka)$ is odd.  
Thus, by Theorem \ref{pro:grad1}, we have $\ha{C} = \rf(\nu)$. \\
 2) Now let $\mu \in K^{(p)}$. By Proposition \ref{pro:gradi3},
if $\ka \in \ha{W}$, then $\Omega_\ka^\mu$ is the set composed of all cogradient paths from $\nu$ to $\ka$ . 
If $\ka \in \ha{W}$, we write $\Omega_\nu^\mu (\ka)$ for the set composed of all paths in $\Omega_\nu^\mu$ which contain $\ka$. \\
By Proposition \ref{pro:gradi3}, each \\
Since each path in $\Omega_\nu^\mu$ contains a single critical simplex, by Theorem \ref{th:gradi2},
 we have $\mu \in \Phi^\infty (\nu)$ if and only if $\sum \{Card(\Omega_\nu^\mu (\ka)) \; | \; \ka \in \ha{W} \}$ is odd. \\
We have $Card(\Omega_\nu^\mu (\ka)) = Card(\Omega_\nu^\ka) \times Card(\Omega_\ka^\mu)$. 
By parity considerations, we have
$\mu \in \Phi^\infty (\nu)$ if and only if: \\
 $\sum \{Card(\Omega_\ka^\mu (\ka)) \; | \; \ka \in \ha{W}$ and $Card(\Omega_\nu^\ka)$ is odd $ \}$ is odd. \\
Thus, $\mu \in \Phi^\infty (\nu)$ if and only if $\sum \{Card(\Omega_\ka^\mu (\ka)) \; | \; \ka \in \ha{C} \}$ is odd. \\
Since $\ha{C} = \rf(\nu)$, and by Theorem \ref{pro:grad1}, $\mu \in \Phi^\infty (\nu)$ if and only if $\mu \in \ti{\rf} (\rf (\nu))$. \\

- We have $\mu \in \ti{\rf} (\rf (\nu))$ if and only if $\sum \{Card(\Omega_\ka^\mu) \; | \; \ka \in \ha{C} \}$ is odd. \\
 We have $\mu \in \ti{\rf} (\rf (\nu))$ if and only if $\sum \{Card(\Omega_\ka^\mu) \; | \; \ka \in \ha{C_o} \}$ is odd. \\
 - We have $\mu \in \Phi^\infty (\nu)$ if and only if $Card(\Omega_\nu^\mu)$ is odd. \\

$Card(\Omega_\nu^\mu (\ka)) = Card(\Omega_\nu^\ka) \times Card(\Omega_\ka^\mu)$. \\
By the above result, we have $Card(\Omega_\nu^\mu (\ka))$ is  odd if and only if $Card(\Omega_\ka^\mu)$ is odd. \\
By parity considerations, we have 
$\mu \in \Phi^\infty (\nu)$ if and only if $\sum \{Card(\Omega_\ka^\mu) \; | \; \ka \in \ha{C} \}$ is odd. \\
}
\end{proof}

As an illustration of the previous result, let us once again consider the example given in  Fig. \ref{fig:Flow}. 
The reference and coreference maps $\rf$  and $\tirf$ of $\ms$ are shown in (e) and (f), using the same coloring conventions as in the preceding examples. 
We have $\rf(a) = \rf(b) + \rf(c) = g+i$. Also we can check that $\tirf \circ \rf (a) = \tirf(g+i) = \ov{\Phi} (a)$. 

Thus, by Theorem \ref{th:gradi4}, the flow  $\ov{\Phi}$ may be decomposed into the two maps  $\ti{\rf}$ and $\rf$. Observe that if $\nu \in K^{(p)}$, 
then each step of the propagation of the flow from $\nu$ involves faces of dimensions $p-1$, $p$, and $p+1$. 
On the other hand, the propagation from $\nu$ to obtain $\mu = \rf(\nu)$ concerns faces of dimensions 
$p$ and $p+1$, while the propagation to get $\ti{\rf} (\mu)$ concerns faces of dimensions $p$ and $p-1$.
Therefore, the decomposition of the flow  $\ov{\Phi}$ with $\ti{\rf}$ and $\rf$ may also be interpreted  as a 
separation of dimensions.

With regard to the computational aspects, we note that the flow can be efficiently obtained if the maps $\rf$ and $\ti{\rf}$
have been precomputed. More specifically, the values of $\ov{\Phi}(\nu)$ for  different faces $\nu$ can be obtained without
redundant computations. 

 Classically, the flow complex of a complex $K$ is defined with the vector space 
 $F_p = \{c \in K[p] \; | \; \Phi_p(c) = c \}$, see Section 7 of \cite{For98a} and Section 8.2 of \cite{Sco19}.  
 We have $F_p = \{c \in K[p] \; | \; \ov{\Phi}_p(c) = c \}$. Also, we observe that: 
 
- If $c \in F_p$, then by Theorem \ref{th:gradi4}, $\ti{\rf}_p (\rf_p (c)) = c$. Therefore we have $c = \ti{\rf}_p (d)$, with $d \in \ha{W} [p]$. 
 It follows that $c \in \ov{O} [p]$. 
 
- If $c \in \ov{O} [p]$, then $c = \ti{\rf}_p (d)$ with $d \in \ha{W} [p]$. By Theorem \ref{th:gradi4}, 
 we have $\ov{\Phi}_p(c) = \ti{\rf}_p (\rf_p (c)) = \ti{\rf}_p (\rf_p (\ti{\rf}_p (d)))$. By Proposition \ref{pro:ext8},
 we derive  $\ov{\Phi}_p(c) = \ti{\rf}_p (d) = c$. Thus $c \in F_p$.  \\
 Therefore, as a direct consequence of the previous theorem,  we obtain 
 the equivalence
 between a flow complex and an extension complex. \\
 
  
 \begin{corollary}
 If $\ms$ is a Morse sequence, then we have: \\
 \hspace*{\fill}
  $\ov{O} [p] = \{c \in K[p] \; | \; \ov{\Phi}(c) = c \}$ and $\un{O} [p] = \{c \in K[p] \; | \; \un{\Phi}(c) = c \}$. 
   \hspace*{\fill} \\
 
 \end{corollary}
 
 \ELIMINE{
 ************************* \\
 Let $\ms(K)$ be a Morse sequence.  \\
 $\Phi_p:$ $K[p] \rightarrow  K[p]$  \\
 For each $\nu \in K^{(p)}$, $\Phi_p (\nu) = \nu + \pl_{p+1}(V_p(\nu)) + V_{p-1}( \pl_p(\nu))$. \\
 
 Suppose $\nu \in \ha{W}$. Then $\Phi_p (\nu) = \nu + V_{p-1}( \pl_p(\nu))$. Thus, if $\mu \in \Phi_p (\nu)$, we have either 
 $\mu = \nu$ or $\mu \in  \ov{W}$. Also if $\mu \in \ov{W}$, we have $\Phi_p (\mu) = \mu + V_{p-1}( \pl_p(\mu))$.
 
 Now, let $f:$ $K[p] \rightarrow  K[p]$ be the linear map such that,
 for each $\nu \in K^{(p)}$, $f_p (\nu) = \nu +  V_{p-1}( \pl_p(\nu))$. \\
 
 By the preceding observation, if $\nu \in \ha{W}$, we have $\ov{\Phi}_p (\nu) = \ov{f}_p(\nu)$. \\
 
 If $\nu \in \ha{W}$, then we have $\ov{f}_p(\nu) = \ti{\rf}(\nu)$. \\

 Let $\nu, \mu \in K^{(p)}$.
Suppose $\nu \in \ha{W}$. Then, $\Omega_\nu^\mu$ is necessarily a cogradient path. By Theorem \ref{th:gradi2},
Theorem \ref{pro:grad1}, and  Definition~\ref{def:ext}, we obtain 
$\ov{\Phi}(\nu) = \ti{\rf}(\nu)$.
 
 }

\section{Conclusion}

In the final chapter of his gentle introduction to discrete Morse theory \cite{Sco19}, Nicholas Scoville  wrote,
``As is now apparent, the idea behind discrete Morse theory is extremely
simple: every simplicial complex can be broken down (or, equivalently,
built up) using only two moves: 1) perform an elementary collapse, and
2) remove a facet.'' 
In this paper, we just start from this basic fact for an alternative approach to this theory.  
A Morse sequence is simply made of two operators which provide
a direct introduction to discrete Morse theory with two fundamental aspects:
 homotopy with the expansion operator, and homology
with the filling operator. 

A key aspect of a Morse sequence is that it gives rise to the maps $\rf$ and $\tirf$.
These two maps, along with their dual $\corf$ and $\ticorf$, contain the crucial homological information of the sequence.

The reference map $\rf$ assigns, to each simplex, a set of critical simplices. 
This leads us to a chain complex, the critical complex, which corresponds precisely
to the classical definition of such a complex. The map $\rf$
is a chain map: it allows to carry out the homology of the original complex to the
critical complex.
 
The extension map $\tirf$ assigns, to each critical simplex, a set of simplices.
It is noteworthy  that $\rf$ and $\tirf$ are ``homology inverses".
This fact allows us to recover a fundamental property, that is, the equality of
homology between a complex and its critical complex. Also, extension maps
leads directly to a chain complex which fully characterizes a flow complex.

Arranged Morse sequences constitute a special class of sequences which contains different kinds of skeletons
of the complex. Any Morse sequence may be transformed to an arranged sequence by a simple reordering
which preserves the gradient vector field.
The skeleton sequence which is included in an arranged sequence allows us to derive some crucial properties of the map $\rf$. 
Furthermore, it leads to a result which provides another aspect of the fundamental collapse theorem of discrete Morse theory.

Since reference and extension maps can be computed using simple operators, 
these notions can have a practical impact on calculating different homological characteristics of a complex.
In future work \cite{BN25}, we plan to leverage this property to develop new algorithms for topological data analysis and digital image segmentation
\cite{robins2011theory,Delgado15,de2015morse}. 
Additionally, we will establish connections between Morse sequences and concepts from Mathematical Morphology, such as watersheds \cite{CoustyBNC10,CoustyBCN14,COMIC201643,BertrandBN23}.


\bmhead{Acknowledgements}
The author wishes to express his thanks to Laurent Najman
for his active interest in this paper and 
for stimulating conversations.

\bibliographystyle{splncs04}
\bibliography{biblio}

\appendix

\section{Discrete vector fields}
\label{app:vector}

We first recall the definitions of a discrete vector field and a $p$-gradient path, see Definitions 2.43 and 2.46 of \cite{Sco19}.
See also \cite{7034732,GONZALEZLORENZO20171} where a generalization of gradient vector fields is introduced, 
which allows cycles under a certain algebraic condition. 

Let $K$ be a complex and $V$ be a set of pairs $(\sigma,\tau)$, with $\sigma, \tau \in K$ and
$\sigma \in \partial \tau$.
We say that $V$ is a  \emph{(discrete) vector field on $K$}
if each simplex of $K$ is in at most one pair of $V$.
We say that $\sigma \in K$ is \emph{critical for $V$} if $\sig$ is not in a pair of $V$.

Let $V$ be a vector field on a complex $K$.
A \emph{($p$-)gradient path in $V$ (from $\sig_0$ to $\sig_k$)}
is a sequence
$\pi = \langle \sig_0,\tau_0,\sig_1,\tau_1,...,\sig_{k-1},\tau_{k-1}, \sigma_{k}\rangle$,
with $k \geq 0$, composed of faces $\sig_i \in K^{(p)}$, $\tau_i \in K^{(p+1)}$
such that, for all $i \in [0,k-1]$,
$(\sig_i,\tau_{i})$ is in~$V$,
$\sigma_{i+1} \subset \tau_i$, and $\sig_{i+1} \not= \sig_i$.
This sequence $\pi$ is said to be \emph{trivial} if $k =0$, that is, if 
$\pi = \langle \sig_0 \rangle$; otherwise, if $k \geq 1$, we say that $\pi$ is \emph{non-trivial}.
Also, the sequence $\pi$ is \emph{closed} if $\sig_0 = \sig_k$.
We say that a vector field $V$ is \emph{acyclic} if $V$ contains
no non-trivial closed $p$-gradient path. 

Now, let $\ms$ be a Morse sequence on $K$. Then, the gradient vector field of $\ms$ is clearly a vector field.
Also, a path is a gradient path in this vector field if and only if it is a gradient path in $\ms$. 

In the sequel of this section,
we want to emphasize that a Morse sequence may be seen as an alternative way to represent the
gradient vector field of an arbitrary discrete Morse function. A classical result of discrete Morse t\textit{}heory states that
a discrete vector field $V$ is the gradient vector field of a
discrete Morse function if and only if $V$ is acyclic (Theorem 2.51 of \cite{Sco19}). Thus, in order to achieve this goal,
it is sufficient to establish the equivalence between gradient vector fields of Morse sequences and acyclic vector fields.
In fact, this equivalence  is a direct consequence of Theorem 11.9 of \cite{koz20} which is formulated in terms of linear extensions associated to a poset.
We give hereafter a proof which is based on the notion of a maximal gradient path. Such a path allows us to extract, in the top dimension of a complex $K$,
either a critical simplex or a free pair for $K$ (Lemma \ref{lem:DVF1}). This formalizes  an incremental
deconstruction of a complex, which is usually 
given with certain Morse functions, see Remark 13 of \cite{Ben17}.

Let $V$ be a vector field on $K$ and 
let $\pi = \langle \sig_0,\tau_0,...,\sig_{k-1},\tau_{k-1}, \sigma_{k} \rangle$ be a $p$-gradient path in $V$.
We say that a pair of simplices $(\eta,\nu)$ is an \emph{extension of~$\pi$ (in $V$)}
if $\langle \eta,\nu,\sig_0,\tau_0,...,\sig_{k-1},\tau_{k-1}, \sigma_{k} \rangle$ or if $\langle \sig_0,\tau_0,...,\sig_{k-1},\tau_{k-1}, \sigma_{k}, \nu, \eta \rangle$ is a $p$-gradient path in $V$.
We say that $\pi$ is \emph{maximal (in $V$)} if $\pi$ has  no extension in~$V$. 
If $V$ is acyclic, it can be checked that, for any 
$p \geq 0$, there exists a maximal $p$-gradient path in $V$.
To see this point, we can pick an arbitrary (possibly trivial) $p$-gradient path and 
extend it iteratively with extensions. If $V$ is acyclic, we obtain 
a maximal path after a finite number of extensions. \\

 \begin{lemma} \label{lem:DVF1}
Let $V$ be an acyclic vector field on $K$, with $dim(K) = d$.
Then, at least one of the following holds: 
\begin{enumerate}
\item There exists a facet $\tau$ of $K$, with $dim(\tau) = d$, that is critical for $V$. 
\item There exists a pair $(\sig,\tau)$ in $V$, with $dim(\tau) = d$, that is a free pair for $K$.
\end{enumerate}
\end{lemma}

\begin{proof}
If $K$ has a $d$-face that is critical for $V$, then we are done. 
Suppose there is no such faces in $K$. If $d= 0$, then all the $0$-faces of $K$ are critical, thus we must have $d \geq 1$.
Let $\tau$ be an arbitrary $d$-face of $K$.
Since $\tau$ is not critical, 
there exists a pair $(\sig,\tau)$ that is in $V$.
Since $d \geq 1$, there is a face $\sig' \in K$
such that $\pi' = \langle \sig,\tau,\sig' \rangle$ is a $(d-1)$ gradient path in~$V$.
By iteratively extending $\pi'$ with extensions we obtain 
a maximal $(d-1)$-gradient path in $V$ that is non-trivial.
Let $\pi = \langle \sig_0,\tau_0,...,\sig_{k-1},\tau_{k-1}, \sigma_{k} \rangle$ be such a path, we have $k \geq 1$.
If $(\sigma_0,\tau_0)$ is a free pair for $K$, then we are done.
Otherwise, $\sig_0$ must be a subset of a $d$-simplex  $\nu$, with $\nu \not= \tau_0$. By our hypothesis $\nu$ is not critical for $V$. Since $\nu$ is a facet
for $K$, there must exist a $(d-1)$-simplex $\eta$, $\eta \not= \sig_0$, such that $(\eta,\nu)$ is in $V$.
In this case, the path $\pi' = \langle \eta, \nu, \sigma_0, \tau_0,...,\sig_{k-1},\tau_{k-1}, \sigma_{k} \rangle$ would be a $(d-1)$-gradient path in~$V$.
Thus, the path $\pi$ would not be maximal, a contradiction:
the pair $(\sigma_0,\tau_0)$ must be a free pair for $K$. \end{proof}

\begin{theorem} \label{th:DVF}
Let $K$ be a simplicial complex. A vector field $V$ on $K$ is acyclic if and only if $V$ is the gradient vector field of a
Morse sequence on $K$.
\end{theorem}

\begin{proof}
i) Let $\ms = \langle \emptyset = K_0,...,K_i,...,K_l = K  \rangle$ be a Morse sequence on $K$, and let $V$ be the gradient vector field of $\ms$.
For each $\nu \in K$, let $\rho(\nu)$ be the index $i$ such that 
$\nu \in K_i$ and $\nu \not\in K_{i-1}$. 
Now, let $\pi = \langle \sig_0,\tau_0,\sig_1,\tau_1,...,\sig_{k-1},\tau_{k-1}, \sigma_{k}\rangle$, $k \geq 1$, be a non-trivial $p$-gradient path in $V$. 
For all $i \in [0,k-1]$, $(\sig_i,\tau_{i})$ is in~$V$, thus 
$\rho(\sig_i) = \rho(\tau_i)$. Since $\sigma_{i+1} \subset \tau_i$ and since 
$\ms$ is a filtration, we have $\rho(\sig_{i+1}) \leq \rho(\tau_i)$. 
Since $\sig_{i+1} \not= \sig_i$ the pair $(\sig_{i+1},\tau_i)$ is not a regular pair for 
$\ms$, thus we have $\rho(\sig_{i+1}) < \rho(\tau_i)$.
It follows that, for all $i \in [0,k-1]$, we have $\rho(\sig_{i+1}) < \rho(\sig_i)$.
This gives $\rho(\sig_{k}) < \rho(\sig_0)$. It means that $\sig_{k} \not= \sig_{0}$,
and the path $\pi$ cannot be closed. Thus the vector field $V$ is acyclic. \\
ii) 
Let $V$ be an acyclic vector field on $K$, with $dim(K) = d$. \\
1) Suppose there exists a facet $\tau$ of $K$, with $dim(\tau) = d$, that is critical for $V$. \\
Let $K' = K \setminus \{\tau\}$ and $V' = V$.
Then, the set $V'$ is also an acyclic vector field on $K'$. \\
2) Suppose there exists a pair $(\sig,\tau)$ in $V$, with $dim(\tau) = d$, that is a free pair for $K$.
Clearly, the set $V' = V \setminus \{(\sigma, \tau) \}$ is also an acyclic vector field on 
$K' = K \setminus \{\sigma, \tau \}$.\\
By 1), 2), and by Lemma \ref{lem:DVF1}, we can build inductively two sequences $\overleftarrow{W}= \langle K = K_0,...,K_k = \emptyset \rangle$
and $\langle V = V_0,...,V_k = \emptyset \rangle$ such that, for each $i \in [1,k]$: \\
- either $K_i$ is an elementary perforation of $K_{i-1}$ and $V_i = V_{i-1}$, \\
- or $K_i = K_{i-1} \setminus \{\sigma, \tau \}$ is an elementary collapse of $K_{i-1}$ and $V_i = V_{i-1} \setminus \{(\sigma, \tau) \}$.\\
By considering the inverse of $\overleftarrow{W}$ we obtain
the sequence $\ms = \langle K_k= K'_0,...,K'_k = K_0 \rangle$, which is such that, for each $i \in [1,k]$,
either $K'_i$ is an elementary expansion of $K'_{i-1}$, or $K'_i$ is an elementary filling of $K'_{i-1}$.
In other words, $\ms$ is a Morse sequence on $K_0 = K$; the gradient field of $\ms$ is precisely $V,$ as required. 
\end{proof}

\section{Morse functions and Morse sequences}
\label{app:functions}

Discrete Morse theory is classically introduced through the concept of a discrete Morse function. In this section we 
show that it is possible, in a straightforward manner,  to make a link between Morse sequences and these functions. 

We first introduce the notion of a Morse function on a 
Morse sequence $\ms$. \\

\begin{definition}
Let $\ms(K)$ be a Morse sequence, let $\diamond \ms = \langle \ka_1, \ldots, \ka_k \rangle$.
We say that a map $f \colon K \to \bb{Z}$ is a \emph{Morse function on $\ms$}  if $f$ satisfies the two conditions:
\begin{enumerate}
\item  If $\ka_i$ is critical for $\ms$ and $\nu \in \partial(\ka_i)$, then $f(\ka_i) > f(\nu)$.
\item If $\ka_i = (\sig,\tau)$ is regular for $\ms$, then $f(\sig) \geq f(\tau)$.     \\
\end{enumerate}
\end{definition}

Now, we can check that the following definition of a Morse function on a simplicial complex $K$ is equivalent to the classical one \cite{For98,For02}.

Let $K$ be a simplicial complex and let $f \colon K \to \bb{Z}$ be a map on $K$. Let $V$ be the set of all pairs $(\sigma, \tau)$,
with $\sig, \tau \in K$, such that $\sig \in \partial(\tau)$ and $f(\sig) \geq f(\tau)$.
If each $\nu \in K$ is in at most one pair in $V$,
we say that $f$ is  a \emph{Morse function  on~$K$}, and $V$ is the \emph{gradient vector field of $f$}.
We say that two Morse functions on $K$ are \emph{equivalent} if they have the same gradient vector field. 

Let $f$ be  a Morse function  on $K$, and $V$ be the gradient vector field of~$f$.
From the above definition, the set $V$ is a discrete vector field on $K$.
If $\pi = \langle \sig_0,\tau_0,\sig_1,\tau_1,...,\sig_{k-1},\tau_{k-1}, \sigma_{k}\rangle$ is a $p$-gradient path in $V$,
we have $f(\sig_i) \geq f(\tau_i)$, and also $f(\tau_i) > f(\sig_{i+1})$.
Thus, $f(\sig_0) > f(\sig_k)$ whenever $k \geq 1$. It means that $V$ contains
no non-trivial closed $p$-gradient path. In other words, we have the classical result:\\

 \begin{proposition} \label{pro:mf1}
If $f$ is a Morse function on $K$, then the gradient vector field of $f$ 
is an acyclic vector field.\\
 \end{proposition}

Let $\ms$ be a Morse sequence on $K$. 
We see that a Morse function on $\ms$ is indeed a Morse function
on $K$, the gradient vector field of this Morse function is
precisely the gradient vector field of $\ms$.
Conversely, by Proposition \ref{pro:mf1} and by Theorem~\ref{th:DVF}, if $f$ is a Morse function on $K$, then there exists a
Morse sequence $\ms$ on $K$ which has the same gradient vector field as $f$. It is easy to check that~$f$ is also a Morse function 
on $\ms$. This leads us to the following result. \\

 \begin{theorem} \label{th:mf3}
If $f$ is a Morse function on $K$, then there exists a Morse sequence~$\ms$ 
on $K$ such that $f$ is a Morse function on $\ms$.
Furthermore, any Morse function on $\ms$ is equivalent to $f$.\\
 \end{theorem}

We introduce hereafter a particular kind of Morse 
function. 
Since a Morse sequence is a filtration, the following
function $f$ is indeed a Morse function on~$\ms$. \\

\begin{definition}
Let $\ms$ be a Morse sequence on $K$ and let $\diamond \ms = \langle \ka_1, \ldots, \ka_k \rangle$.
The \emph{canonical Morse function of $\ms$}  is the map 
$f \colon K \to \bb{Z}$ such that:
\begin{enumerate}
\item If $\ka_i$
is critical for $\ms$, then $f(\ka_i) =i$.
\item If $\ka_i = (\sig,\tau)$
is regular for $\ms$, then $f(\sig) = f(\tau) =i$. \\
\end{enumerate}
\end{definition}

As a consequence of Theorem \ref{th:mf3}, any Morse function 
on $K$ is equivalent to a canonical Morse function. 

We note that a canonical Morse function $f$ is \emph{flat}, that is,
we have  $f(\sig) = f(\tau)$ whenever $(\sig,\tau)$ is in the gradient vector field of $f$
(Definition 4.14 of \cite{Sco19}).
Also~$f$ is \emph{excellent}, that is, all values of the critical simplices
are distinct (Definition  2.31 of \cite{Sco19}).
In fact, a canonical Morse function has the three properties
which define a basic Morse function (see \cite{Ben16} and also Definition 2.3 of \cite{Sco19}).

Let $f \colon K \to \bb{Z}$ be a map on $K$. 
We say that $f$ is a \emph{basic Morse function}
if $f$ satisfies the properties:
\begin{enumerate}
\item \emph{monotonicity}: we have $f(\sig) \leq f(\tau)$ whenever $\sig \subseteq \tau$;
\item \emph{semi-injectivity}: for each $i \in \bb{Z}$, the cardinality of $f^{-1}(i)$ is at most 2;
\item \emph{genericity}: if $f(\sig) = f(\tau)$, then either  $\sig \subseteq \tau$ or $\tau \subseteq \sig$. 
\end{enumerate}
We observe that, if $f$ is a basic Morse function on $K$,
then we can build a Morse sequence $\ms$
if we pick the simplices of $K$ by increasing values of $f$. 
For each~$i$, $f^{-1}(i)$ gives a critical simplex if 
the cardinality of $f^{-1}(i)$ is one, and
$f^{-1}(i)$ gives a regular pair if 
the cardinality of $f^{-1}(i)$ is two.

Let $f$ and $g$  be two basic Morse functions on $K$. 
We say that $f$ and $g$ are \emph{strongly equivalent}
if $f$ and $g$ induce the same order on $K$. That is, we have 
$f(\sig) \leq f(\tau)$ if and only if 
$g(\sig) \leq g(\tau)$.

With the above scheme for building a Morse sequence from 
a basic Morse function, we obtain the following result. \\

\begin{proposition}
Let $f$ be a basic Morse function on $K$.
There exists one and only one Morse sequence $\ms$ such that 
the canonical Morse function of $\ms$ is strongly equivalent to $f$.
\end{proposition}

\end{document}